\documentclass[12pt,authoryear,review]{elsarticle}
\makeatletter

  \def\ps@pprintTitle{%
 \let\@oddhead\@empty
 \let\@evenhead\@empty
 \def\@oddfoot{\centerline{\thepage}}%
 \let\@evenfoot\@oddfoot}
\makeatother
\usepackage{amsmath,amssymb,amsfonts}
\usepackage{graphicx}
\usepackage{subfigure}
\usepackage{color}
\usepackage{bigstrut}
\usepackage{multirow}
\usepackage{pdflscape}
\usepackage{ulem}
\usepackage{dsfont}
\usepackage[dvipsnames]{xcolor}

\usepackage{hhline}
\usepackage{epstopdf}

\usepackage{mathrsfs}

\journal{TBA}

\newtheorem{theorem}{Theorem}

\newtheorem{proposition}[theorem]{Proposition}

\newdefinition{remark}{Remark}
\newdefinition{assumption}{Assumption}
\newproof{proof}{Proof}
\newproof{pot}{Proof of Theorem \ref{thm2}}

\newcommand{\binfty}{{\boldsymbol \infty}}

\newcommand{\EE}{{\mathbb E}}

\newcommand{\Acal}{{\mathcal A}}
\newcommand{\bmu}{{\boldsymbol\mu}}
\newcommand{\bkappa}{{\boldsymbol\kappa}}
\newcommand{\bP}{{\boldsymbol P}}
\newcommand{\bg}{{\boldsymbol g}}
\newcommand{\bo}{{\boldsymbol o}}
\newcommand{\bp}{{\boldsymbol p}}
\newcommand{\bSigma}{{\boldsymbol\Sigma}}
\newcommand{\bW}{{\boldsymbol W}}
\newcommand{\btheta}{{\boldsymbol\theta}}
\newcommand{\bsigma}{{\boldsymbol\sigma}}
\newcommand{\balpha}{{\boldsymbol\alpha}}
\newcommand{\bbeta}{{\boldsymbol\beta}}
\newcommand{\bzeta}{{\boldsymbol\zeta}}
\newcommand{\bb}{{\boldsymbol b}}
\newcommand{\ba}{{\boldsymbol a}}
\newcommand{\bA}{{\boldsymbol A}}
\newcommand{\bB}{{\boldsymbol B}}
\newcommand{\bC}{{\boldsymbol C}}
\newcommand{\bwC}{\widetilde{\boldsymbol C}}
\newcommand{\bD}{{\boldsymbol D}}
\newcommand{\bmD}{{\boldsymbol{\mathfrak{D}}}}

\newcommand{\bmd}{{\boldsymbol{\mathfrak{f}}}}
\newcommand{\bE}{{\boldsymbol E}}
\newcommand{\bF}{{\boldsymbol F}}
\newcommand{\bG}{{\boldsymbol G}}
\newcommand{\bGamma}{\boldsymbol \Gamma}
\newcommand{\bH}{{\boldsymbol H}}
\newcommand{\bI} {\boldsymbol I}

\newcommand{\bmL}{{\boldsymbol {\mathcal L}}}
\newcommand{\bM}{{\boldsymbol {M}}}
\newcommand{\bmM}{{\boldsymbol {\mathcal M}}}
\newcommand{\bmN}{{\boldsymbol {\mathcal N}}}
\newcommand{\bzero}{{\boldsymbol{0}}}
\newcommand{\bQ}{{\boldsymbol Q}}
\newcommand{\bmQ}{\boldsymbol {\mathcal Q}}
\newcommand{\bS}{{\boldsymbol{S}}}
\newcommand{\bmS}{{\boldsymbol {\mathcal S}}}
\newcommand{\bX}{{\boldsymbol {\mathcal{X}}}}
\newcommand{\bnu}{{\boldsymbol\nu}}

\newcommand{\bw}{\boldsymbol w}

\newcommand{\bZ}{{\boldsymbol Z}}
\newcommand{\bmZ}{{\boldsymbol{\mathfrak{Z}}}}
\newcommand{\bz}{{\boldsymbol z}}
\newcommand{\bbq}{{\boldsymbol q}}
\newcommand{\bbs}{{\boldsymbol{s}}}
\newcommand{\be}{{\boldsymbol{e}}}

\newcommand{\toe}[1]{\boldsymbol{:e}^{#1}\boldsymbol{:}}

\newsavebox\CBox
\newcommand\hcancel[2][0.5pt]{%
  \ifmmode\sbox\CBox{$#2$}\else\sbox\CBox{#2}\fi%
  \makebox[0pt][l]{\usebox\CBox}%
  \rule[0.5\ht\CBox-#1/2]{\wd\CBox}{#1}}

\begin{document}

\begin{frontmatter}

\title {\textbf{Trading Co-Integrated Assets with Price Impact}\tnoteref{t1} \\[0.5em]
\textit{Mathematical Finance, Forthcoming}
}
\tnotetext[t1]{ SJ would like to thank NSERC and GRI for partially funding this work. We are grateful to two anonymous referees and the editor (J. Detemple) for valuable comments that improved this paper.  }

\author[author1]{\'Alvaro Cartea}
\ead{alvaro.cartea@maths.ox.ac.uk}
\author[author2]{Luhui Gan}
\ead{luke.gan@mail.utoronto.ca}
\author[author2]{Sebastian Jaimungal}
\ead{sebastian.jaimungal@utoronto.ca}
\address[author1] {Mathematical Institute, University of Oxford, Oxford, UK\\  Oxford-Man Institute of Quantitative Finance, Oxford, UK}
\address[author2] {Department of Statistical Sciences, University of Toronto, Toronto, Canada}

\begin{abstract}
Executing a basket of co-integrated assets is an important task facing investors.  Here, we show how to do this accounting for the informational advantage gained from assets within and outside the basket, as well as for the permanent price impact of market orders (MOs) from all market participants, and the temporary impact that the agent's MOs have on prices. The execution problem is posed as an optimal stochastic control problem and we demonstrate that, under some mild conditions, the value function admits a closed-form solution, and prove a verification theorem. Furthermore,
we use data of five stocks traded in the Nasdaq exchange to estimate the model parameters and use simulations to illustrate the performance of the strategy. As an example,  the agent liquidates a portfolio consisting of shares in INTC\footnote{INTC: Intel Corporation.} and SMH.\footnote{SMH: Market Vectors Semiconductor ETF.} We show that including the  information provided by three additional assets (FARO, NTAP, ORCL)\footnote{FARO: FARO Technologies. NTAP: NetApp. ORCL: Oracle Corporation.} considerably improves the strategy's performance;  for the portfolio we execute, it outperforms the multi-asset version of Almgren-Chriss  by approximately $4$ to $ 4.5 $ basis points.
\end{abstract}


\end{frontmatter}

\section{Introduction}

How to optimally execute a large position in an individual stock has been a topic of intense academic and industry research during the last few years. In contrast,  there is scant work on the joint execution of large positions in multiple assets.  One of the  early papers on optimal execution is by  \cite{almgren2001optimal} who consider a discrete-time model where the strategy employs market orders  (MOs) only. Extensions of their work,  where the agent employs MOs and/or limit orders, include \cite{almgren2012optimal}, \cite{kharroubi2010optimal}, \cite{Gueant2012}, \cite{forsyth2012optimal}, \cite{jaimungal2013optimal}, \cite{guilbaud2013optimal},  and \cite{cartea2014optimal}. In the extant literature, if the agent liquidates a portfolio of different assets, these are considered to be correlated, but do not include co-integration, nor do they include the market impact of the order flow from other market participants.  This paper fills this gap.   We  show how an agent  executes a basket of assets employing a framework that models  the price impact of  order flow,  and  employs the information provided by the co-integration factors that drive the joint dynamics of prices -- which may include other assets she is not trading in.

In our framework, the agent's  MOs have both temporary and permanent price impact. Temporary impact results from the agent's MOs walking the limit order book (LOB), and permanent impact results from   one-sided trading pressure exerted on prices. In contrast to most of the literature (\cite{cartea2015incorporating} and \cite{cartea2014closed} being two notable exceptions), here, MOs of other market participants are treated in the same way as the agent's order: market buy orders exert upward  pressure on prices, and market sell orders downward pressure on prices.  Furthermore, order flow in one asset may impact  the   prices of co-integrated assets.  This cross-effect  is partly caused by trading algorithms that take positions based on the co-movements of assets. Such strategies induce  co-movement in order flow and  liquidity displayed in the LBOs of the co-integrated assets.

In our setup,  permanent impact of order flow is linear in the speeds of trading of all market participants (including the agent), and temporary impact   is also linear in the agent's speed of trading. We focus on the execution problem where the agent liquidates shares in $m$ assets and employs information from  a collection of $n\geq m$ co-integrated assets. The agent maximizes terminal wealth and penalizes deviations from an inventory-target schedule. This scenario appears in many applications in practice. For example, agency traders are often faced with liquidating a basket of Eurodollar\footnote{Recall that Eurodollar futures are futures contracts on time deposits  denominated in USD, but held in a non-US country.} futures of consecutive maturities. These contracts are highly co-integrated, and not simply correlated, see the discussion in \cite{almgren2013executionFixed}.

Our setup is related to that of \cite{JOFI:JOFI12080} in which the authors optimize the discounted, and penalized, future expected excess returns in a discrete-time, infinite-time horizon problem. In their model, prices contain an unpredictable martingale component, and an independent stationary predictable component. The penalty is imposed to account for a version of temporary price impact similar to walking the LOB, and they include a permanent price impact which reverts to zero if there are no trades. \cite{AlphaPredictors} numerically study a continuous-time, finite horizon, version of \cite{JOFI:JOFI12080}, and account for crossing the spread or posting limit orders. Our approach differs in five main aspects: (i)   our setup is in continuous-time, (ii) the execution horizon is finite, (iii) the agent solves an execution problem where prices are co-integrated (rather than having independent predictable components), (iv) the agent's MOs have permanent and temporary impact, and (v) the MOs of other market participants also have permanent price impact. Moreover, we provide analytic characterizations of the solution to the execution problem.

To illustrate the performance of the strategy we calibrate model parameters to five stocks  (INTC, SMH, FARO, NTAP, and ORCL) traded on the Nasdaq exchange  and run simulations for variations of the strategy including different levels of urgency and inventory-target schedules, including/excluding a speculative component which allows repurchases of shares. As benchmark we use the multi-asset version of the Almgren-Chriss (AC) strategy where  the agent models the correlation between the assets in the basket, but  does not model co-integration or employ additional information from other assets.   The agent liquidates a basket consisting of 4,600 shares of INTC and 900 shares of SMH  which corresponds to 1\% and 4\% of traded volume over the one hour in which execution occurs.

Additional information from other co-integrated stocks considerably boosts the performance of the strategy. For example,  if the level of urgency required by the agent to liquidate the portfolio is high (resp. low)  the strategy outperforms  AC by 4 (resp. 4.5) basis points.  This improvement over AC is due to the quality of the information provided by the co-integrated assets, and  due to a speculative component of the strategy which allows the agent to repurchase shares during the liquidation horizon to take advantage of price signals. If the agent is not allowed to speculate, i.e., cannot repurchase shares, the relative savings compared to AC, depending on the level of urgency, are between 2.5 to 3.5 basis points.

Finally, we also illustrate how the strategy performs when the agent  has access to only one trading day of data, thus parameter estimates are incorrect. We show that the performance of the strategy is broadly the same as that resulting from that when the agent has enough data to obtain correct parameter estimates.

Our model is also related to the literature on pairs trading in that the agent's strategy benefits from co-integration in asset prices.  For example, \cite{mudchanatongsuk2008optimal} model the log-relationship between a pair of stock prices as an Ornstein-Uhlenbeck process and use this to formulate a trading strategy.  More recently, \cite{leung2015optimal} study the optimal timing strategies for trading a mean-reverting price spread, see also  \cite{lei2015costly}, and  \cite{ngo2014optimal}. Finally, the work of \cite{yan2012dynamic} develops an optimal  portfolio strategy for a pair of co-integrated assets. This is generalized to multiple co-integrated assets in \cite{CarJaiCointegration}, and \cite{lintilhac2016model}.

The remainder  of this paper is structured as follows. Section    \ref{sec:Model} presents the model for the co-integrated prices and poses the liquidation problem solved by the agent. Section \ref{sec:solving_the_control_problem} presents the dynamic programming equation and shows the optimal liquidation speeds. Section \ref{sec:Numerical-Example} discusses the Nasdaq exchange data employed to estimate  the co-integrating factor of prices, and illustrates the performance of the strategy under different assumptions. Section \ref{sec: concs} concludes and  proofs are collected in the Appendix.

\section{Model}\label{sec:Model}

The investor must liquidate a portfolio of assets and has a time limit to complete the execution. One simple strategy is  to view each stock in the portfolio  independently  and employ a liquidation algorithm designed for an individual stock, see e.g. \cite{almgren2001optimal}, \cite{bayraktar2012liquidation}, \cite{TheBook}. Treating each stock independently  is optimal if the assets in the portfolio do not exhibit any co-movements or dependence.

Here we focus on the general   case where a collection of traded assets co-move. Modelling the joint dynamics provides the investor with better information to undertake the liquidation strategy. Ideally, the information employed in the execution strategy is not limited to the constituents of the portfolio to be liquidated, it includes other assets that improve the quality of the information employed in the algorithm. See for example, \cite{AlSebDam13} who show how to learn from  a collection of assets to trade in a subset of the assets.

The portfolio consists of $m$ assets which are a subset of  the $n$-dimensional  vector $\bP=(\bP_t)_{0\le t\le T}$ of  midprices that the investor employs in the trading algorithm. The midprices are determined by a co-integration factor and the impact of the order flow from all market participants including the investor's orders. Specifically we assume that the midprices satisfy the multivariate stochastic differential equation (SDE)
\begin{equation} \label{eqn: controlled midprices}
d\bP_t = d\bS_t +  \bg(\bo_t)\,dt \,,
\end{equation}
where  $ \bS $ denotes the co-integration component of midprices and satisfies
\begin{equation}\label{eqn: cointegrated midprices}
	d\bS_{t} = \bkappa\,\left(\btheta-\bS_{t}\right)\,dt + \bsigma^{\intercal}\,d\bW_{t}\,.
\end{equation}
Here $\bkappa$ is a $n\times n$ matrix, $\btheta$ is an $n$-dimensional vector,  and $\bsigma^{\intercal}$ is the Cholesky decomposition of the
asset prices' correlation matrix $\bSigma$ (i.e. $\bSigma=\bsigma^{\intercal}\bsigma$), where the operation $^\intercal$ denotes the transpose operator. As usual we work on the  filtered probability space $(\Omega,\mathcal{F},\mathbb{P},\mathbb{F}=(\mathcal{F}_{t})_{0\leq t\leq T})$, and $\bW=(\bW_{t})_{0\le t\le T}$ is an $n$-dimensional Brownian motion with natural filtration $\mathcal{F}_{t}$.

Moreover $\bg(\bo_t)$ represents the effect of order flow $\bo=(\bo_t)_{0\le t\le T}$, with $\bo_t\in\mathbb R^n$,  from all market participants (including the investor's trades) on  midprices, and $\bg:\mathbb R^n \to \mathbb R^n$ is a permanent price impact function. Below we give a more detailed account of the effect of order flow on the midprice dynamics -- for more details see \cite{cartea2015incorporating} who discuss the effect of market order flow on asset prices.

The investor wishes to liquidate the portfolio of $m$ assets over a time  window $[0,T]$ -- the setup for the acquisition problem is similar, so we do not discuss it here. Her initial inventory in each asset is given by the vector $\bQ_{0}\in\mathbb{R}^{m}$ and she must choose the speed at which she liquidates each one of the assets using MOs only.

We denote by $\bnu = (\bnu_t)_{0\leq t \leq T}$ the vector of liquidation speeds, and by $\bQ^\bnu = (\bQ^\bnu_t)_{0\leq t \leq T}$ the vector of (controlled) inventory holding in each asset. The inventory is  affected by how fast she trades  and satisfies
\begin{equation}
d\bQ_{t}^{\bnu}=-\bnu_{t}\,dt\,.\label{eq:inventory_process}
\end{equation}

In our model all MOs have price impact.  We assume that price impact is linear in the speed of trading (see \cite{cartea2015incorporating} for extensive data analysis illustrating this fact) and  treat the order flow of the investor and other market participants  symmetrically.  In particular, we denote other agents' aggregated net trading speed  by $\bmu=(\bmu_t)_{0 \leq t \leq T}\,$, which we assume is Markov\footnote{ We can easily include other factors that drive order flow, as long as the joint process, consisting of the driving factors and order flow itself, is Markov.}  with infinitesimal generator $\mathcal{L}^\bmu$,  and independent\footnote{This independence assumption can also be relaxed.} of the Brownian motion $\bW$. Thus, the price impact of order flow is
\begin{equation}\label{eqn: linear price impact}
\bg(\bo_t) = -\bb\,  \bX^\intercal \, \bnu_{t}\, + \overline{\bb} \,\bmu_t \,,
\end{equation}
where   $\bb$ is the permanent impact  $n\times n$ symmetric matrix and $ \overline{\bb} $ is the permanent impact $ n \times n $ matrix from other agents’ trading activity. $\bX $ is a $m\times n$ matrix with $\bX_{ij}={\mathds 1}_{\{i=j\}}$  and maps the first $m$ elements of an $n$-dimensional vector to an $m$-dimensional vector.
Although   permanent impact from order flow is treated symmetrically, here we separate the agent's impact from that of other participants should we want to focus on either one when analyzing  the strategy.

Therefore,  after inserting \eqref{eqn: linear price impact} in \eqref{eqn: controlled midprices},  the midprice can be expressed as
\begin{equation}\label{eq:price_dynamics}	
\bP_t^\bnu= \bS_t + \bb\,\bX^\intercal\,(\bQ^\nu_t - \bQ_0) + \overline{\bb} \, \bmM_t\,,
\end{equation}
where $ \bmM_t = \int_{0}^{t} \bmu_u du $ and we use the notation $ \bP^\bnu $ to stress that midprices are affected by the investor’s (controlled) speed of trading.

Our model for price dynamics is related to that used in optimal `pairs trading' where a speculative strategy is designed to profit from the movement of a collection of co-integrated assets, see  \cite{yan2012dynamic}, \cite{leung2015optimal}, and  \cite{CarJaiCointegration}. Our work is different   in that the agent's objective is to execute a basket of co-integrated assets, and more importantly,   order flow from all market participants, including the agent's own trades, is explicitly modelled in the price dynamics, and (as discussed below, we account for temporary price impact).

In addition to permanent price impact, the  investor receives  worse than  quoted  midprices because her MOs walk the LOBs. This price impact is temporary and only affects the prices the investor  receives when selling shares. The execution prices are given by
\begin{equation}\label{eq:execution_price-1}
	\tilde{\bP}_{t}^{\bnu}=\bX\,\bP_t^\bnu-\ba\,\bnu_{t}\,.
\end{equation}
$\ba$ is an $m\times m$ positive definite matrix, so the temporary impact is linear in the speed of trading. Without loss of generality, we  assume that the first $m$ coordinates of $\bP_t^\bnu$ correspond to the assets the investor trades.

In this setup, the LOBs recover immediately after the execution of the MOs -- see  \cite{Almgren03}, \cite{AlfonsiFruthSchiedQF2010}, \cite{KharroubiPhamSIAM2010}, \cite{SchiedMAFI2012}, \cite{SchiedAMF13},  \cite{GueantLeHalleMF13} for further discussions and generalizations.

Finally, the investor's cash from  liquidating shares in the $m$ assets is denoted  by $X^{\bnu}=(X^{\bnu})_{0\le t \le T}$ and satisfies the SDE
\begin{equation}\label{eq:wealth_dynamics}
	dX_{t}^{\bnu}=(\bX \,\bP_t^\bnu-\ba\,\bnu_{t})^{\intercal}\,\bnu_{t}\,dt\,.
\end{equation}

\subsection{Performance criteria and value function}
The investor aims at liquidating the portfolio by the terminal date $T$  and   maximizes  expected terminal wealth
while penalizing deviations from a deterministic target inventory  $\bmQ_t:\mathbb R_+\to \mathbb R^m$  satisfying $\bmQ_0 = \bQ_0$ and $\bmQ_T = \bzero$.

 Her performance criteria is
\begin{equation}\label{eqn: performance criteria}
	\begin{split}
 	H^\bnu(t,x,\bp, \bbq, \bmu) = &\mathbb{E}_{t,x,\bp,\bbq,\bmu}\bigg[\,X_{T}^{\bnu}+	(\bP_T^\bnu)^{\intercal}\,\bX ^{\intercal}\,\bQ_{T}^{\bnu}-(\bQ_{T}^{\bnu})^{\intercal}\,\balpha\, \bQ_{T}^{\bnu}  \\
	& \qquad \qquad \qquad -\phi\,\int_{t}^{T}(\bQ_{u}^{\bnu} - \bmQ_u)^{\intercal}\,\tilde{\bSigma}\,(\bQ_{u}^{\bnu}-\bmQ_u)\,du\,\bigg]\,,
	\end{split}
\end{equation}
where the expectation operator $\mathbb E_{t,x,\bp,\bbq,\bmu}[\,\cdot\,]$ represents expectation conditioned on (with a slight abuse of notation) $X_{t}^\bnu=x$, $\bP_{t}^\bnu= \bp$, $\bQ_t^\bnu=\bbq$, and $\bmu_t=\bmu$, and $\tilde{\bSigma}$ is an  $m\times m$ sub-matrix of   the correlation matrix $\bSigma$  corresponding to the $m$ assets that are being traded. Her value function is
\begin{equation}\label{eqn:value function}
H(t,x,\bp, \bbq, \bmu)=\sup_{\bnu\in \Acal}H^\bnu(t,x,\bp, \bbq, \bmu)\,,
\end{equation}
where $\Acal $ is the set of admissible  strategies consisting of  $\mathcal F$-predictable processes such that $\int^T_0\, |\nu_u^i|\,du< +\infty$, $\mathbb P-$a.s., for each asset $i$ the investor is liquidating. The liquidation speeds are not restricted to remain positive -- we return to this point in Section \ref{sec:Numerical-Example} when we analyze the empirical performance of the trading strategy.

 The first term on the right-hand side of the performance criteria \eqref{eqn: performance criteria} is  the terminal cash. The second term represents the cash obtained from liquidating  all remaining shares at the end of the trading window at the price $\bP$. The third term is the market impact costs from liquidating final inventory  which are encoded in the positive definite matrix $\balpha>0$.

Furthermore, the term in the second line of \eqref{eqn: performance criteria}  represents a running inventory-target penalty where $\phi\geq 0$ is a penalty parameter.  This inventory penalty does not affect the investor's revenues, but affects the optimal liquidation rates. When the value of the inventory penalty parameter $\phi$ is high, the strategy is forced to track closely the target $\bmQ_t$. This is similar in spirit to \cite{cartea2014closed} who develop a trading strategy to target VWAP, i.e.,  volume weighted average price, and similar to \cite{bank2015hedging} who study how to target general positions (without any price dynamics).

For example, when   $\bmQ_t=\bzero$ over the execution window,  $\phi$ may be interpreted as an urgency parameter. High values of $\phi$ correspond to the trader wishing to rid herself of more inventory early on. This particular target is justified in a setting where the investor considers model uncertainty -- i.e., she is ambiguity averse. \cite{CDJ} show  that including a running penalty that curbs the strategy to draw down inventory holdings to zero is equivalent to the agent considering alternative models with stochastic drifts. In that setting, the higher the value of $\phi$, the less confident the agent is about the trend of the midprice, so the quicker the strategy executes the shares.

\section{Optimal Portfolio Liquidation}\label{sec:solving_the_control_problem}

In this section we  derive the optimal liquidation rates.
Our first step is to rewrite the control problem using the fundamental price $ \bS $ as a state variable.
Using (\ref{eq:wealth_dynamics}) and integration by parts, the investor's wealth $ X_T^\bnu $ can be written as
\begin{eqnarray} \label{eq:terminal wealth restated}
X_T^\bnu			 & = & \int_0^T (\bX\,\bS_t - \ba\,\bnu_t)^\intercal\bnu_t\, dt - \tfrac{1}{2} (\bQ_T^\bnu - \bQ_0)^\intercal \bX\,\bb\,\bX^\intercal (\bQ_T^\bnu - \bQ_0) \nonumber \\
			 &  & -\bmM_T^\intercal\,\overline{\bb}^\intercal\,\bX^\intercal\,\bQ_T^\bnu + \int_0^T \bmu_t^\intercal \,\overline{\bb}^\intercal \,\bX^\intercal \,\bQ_t^\bnu \, dt\,.
\end{eqnarray}
From (\ref{eq:price_dynamics}), we have
\begin{eqnarray} \label{eq:book value terminal inventory}
	(\bP_T^\bnu)^\intercal\,\bX^\intercal\,\bQ_T^\bnu & = & \bS_T^\intercal \,\bX^\intercal \,\bQ_T^\bnu + (\bQ_T^\bnu - \bQ_0)^\intercal \,\bX\,\bb\,\bX^\intercal\,\bQ_T^\bnu + \bmM_T^\intercal\,\overline{\bb}^\intercal\,\bX^\intercal\,\bQ_T^\bnu\,,
\end{eqnarray}
and using (\ref{eq:terminal wealth restated}) and (\ref{eq:book value terminal inventory}),
the performance criteria can be written as
\begin{equation}
\begin{split}
	H^\bnu(t,x,\bbs, \bbq, \bmu) = & \mathbb{E}_{t,x,\bbs,\bbq,\bmu}\bigg[ \int_0^T (\bX\,\bS_t - \ba\,\bnu_t)^\intercal\bnu_t\, dt %\bS^\intercal_t\,\bX^\intercal\,\bnu_t\, dt
	+ \bS_T^\intercal \,\bX^\intercal \,\bQ_T^\bnu + \int_0^T \bmu_t^\intercal \,\overline{\bb}^\intercal\, \bX^\intercal\, \bQ_t^\bnu \, dt
	  \\
	& \phantom{\mathbb{E}_{t,x,\bp,\bbq,\bmu}\bigg[} +  (\bQ_T^\bnu)^\intercal \,\left(\frac{1}{2} \,\bX\,\bb\,\bX^\intercal - \balpha\right) \bQ_T^\bnu - \tfrac{1}{2}\, (\bQ_0)^\intercal \,\bX\,\bb\,\bX^\intercal\, \bQ_0   \\
	& \phantom{\mathbb{E}_{t,x,\bp,\bbq,\bmu}\bigg[}  -\phi\,\int_{t}^{T}(\bQ_{u}^{\bnu} - \bmQ_u)^{\intercal}\,\tilde{\bSigma}\,(\bQ_{u}^{\bnu}-\bmQ_u)\,du\,\bigg]\,.
\end{split}
\end{equation}

We further simplify the problem by introducing the transformed processes $Y^\bnu=\{Y_t^\bnu\}_{t\ge0}$ and $\bZ=\{\bZ_t\}_{t\ge0}$  through the following equalities
\begin{align}
	Y_{t}^{\bnu} &= \int_0^t (\bX\,\bS_u - \ba\,\bnu_u)^\intercal\bnu_u\, du   + \btheta^{\intercal}\,\bX^{\intercal}\,(\bQ_{t}^{\bnu}-\bQ_{0})\, + \int_0^t \bmu_u^\intercal\, \overline{\bb}^\intercal \,\bX^\intercal\, \bQ_u^\bnu \, du ,\label{eq:trans_Y} \\
	\bZ_{t} &= \bS_{t} - \btheta \,,\label{eq:trans_cointegration_factor}
\end{align}
in which case $\bZ $ and $Y^\bnu$ satisfy the SDEs
\begin{eqnarray}
	d\bZ_{t}  &= &  -\bkappa \,\bZ_t dt + \bsigma^{\intercal}\,d\bW_{t}\,,
	\label{eq:trans_price_dynamic} \\
	dY_{t}^{\bnu} & = & \Big\{ \left(\bX \,\bZ_{t} - \ba\,\bnu_{t}\right)^{\intercal}\,\bnu_{t}\,
	+ \bmu_t^\intercal\, \overline{\bb}^\intercal \bX^\intercal \bQ_t^\bnu \, \Big\} dt\,,
\end{eqnarray}
and the control problem, in the new variables, becomes
\begin{equation}
\begin{split}
&	H(t,y,\bz,\bbq,\bmu) \\
& \quad = \sup_{\bnu\in\mathcal{U}}\mathbb{E}_{t,y, \bz,\bbq,\bmu}\Bigg[\,Y_{T}^{\bnu} +
	\bZ_T^\intercal \,\bX^{\intercal}\,\bQ_{T}^{\bnu} + \btheta^{\intercal}\,\bX^{\intercal}\,\bQ_{0} + (\bQ_{T}^{\bnu})^{\intercal}\, (\tfrac{1}{2} \bX\bb\bX^\intercal - \balpha)\, \bQ_{T}^{\bnu}     \\
&  \phantom{\quad = \sup_{\bnu\in\mathcal{U}}\mathbb{E}_{t,y, \bz,\bbq,\bmu}\Bigg[\,} \; - \tfrac{1}{2} (\bQ_0)^\intercal \bX\bb\bX^\intercal \bQ_0 -\phi\int_{t}^{T}(\bQ_{u}^{\bnu} - \bmQ_u)^{\intercal}\,\tilde{\bSigma}\,(\bQ_{u}^{\bnu} - \bmQ_u)\,du \,\Bigg]\,,
\end{split}
\label{eq:unconst_problem}
\end{equation}
where $\mathcal U$ is the set of admissible  strategies in the new variables.

\subsection{The dynamic programming equation} \label{sec:DPE}

The dynamic programming principle  suggests that the value function   (\ref{eq:unconst_problem})
is the unique classical solution to the   DPE
\begin{equation}
	\partial_{t}H+\mathcal{L}^\bmu H
+ \sup_{\bnu}\left\{\;\mathcal{L}^{\bnu}H\;\right\}
- \phi \,(\bbq-\bmQ_t)^{\intercal}\,\tilde{\bSigma}\,(\bbq-\bmQ_t)   =   0\,,
	\label{eq:HJB_equation}
\end{equation}
subject to the terminal condition
\begin{equation}\label{eq:value_function_terminal_condition}
H(T,y,\bz,\bbq,\bmu)  =  y + \bz^{\intercal} \, \bX ^{\intercal}\,\bbq + \bbq^{\intercal}\,(\tfrac{1}{2} \bX\,\bb\,\bX^\intercal - \balpha)\, \bbq+\btheta^{\intercal}\,\bX ^{\intercal}\,\bQ_{0} - \tfrac{1}{2} \,\bQ_0^\intercal\, \bX\,\bb\,\bX^\intercal \,\bQ_0 \,,
\end{equation}
and where $\mathcal{L}^{\bnu}$ is the infinitesimal generator of the process $(Y^\bnu,\,\bQ^{\bnu},\,\bZ)$, which acts on a smooth function $\varphi$ as follows
\begin{equation} \label{eq:infi_gen}
\begin{split}
&\mathcal{L}^{\bnu}\varphi(t,y,\bz,\bbq,\bmu)  \\
&\quad = \Big\{\!\left(\bX \,\bz - \ba\,\bnu\right)^{\intercal}\,\bnu + \bmu^\intercal \,\overline{\bb}^\intercal\,\bX^\intercal\,\bbq \Big\} \,\partial_{y}\varphi-\bnu^{\intercal}\,\partial_{q}\,\varphi -\bz^{\intercal}\,\bkappa\, \partial_{z}\,\varphi +\tfrac{1}{2}\,\mbox{Tr}\left(\bSigma\,\partial_{zz}\,\varphi\right)\,.
\end{split}
\end{equation}

\begin{proposition}
\textbf{Solving the DPE.} The DPE (\ref{eq:HJB_equation}) admits the solution
\begin{equation}
	H(t,y,\bz, \bbq, \bmu) =  y+\bz^{\intercal}\,\bA(t)\,\bz
+\bz^\intercal\,\bB(t,\bmu)
+\bbq^{\intercal}\,\bC(t)\,\bbq
+\bbq^\intercal\,\bD(t,\bmu)
+\bz^{\intercal}\,\bE(t)\,\bbq+F(t,\bmu)\,,
	\label{eq:ansatz1}
\end{equation}
if there exists unique matrix-valued functions $\bA(t)$ ($n\times n$), $\bB(t,\bmu)$ ($n\times 1$), $\bC(t)$ ($m\times m$), $\bD(t,\bmu)$ ($m\times 1$), $\bE(t)$ ($n\times m$), and  function $F(t,\bmu)$
that satisfy
\begin{enumerate}[(a)]
\item The  matrix Riccati equation
%\begin{subequations}
\begin{equation}
	\dot{\bG} + \bG \,\bM_1 \bG + \bG\, \bM_2 + \bM_2^\intercal \,\bG + \bM_3 = \bzero \,,
	\label{eq:matrix_riccati_block_form}
\end{equation}
with terminal condition $\bG(T) = \left[\begin{smallmatrix} \bzero^{(n,n)} & \bzero^{(n,m)} \\ \bzero^{(m,n)} & \frac{1}{2}\,\bX\,\bb\,\bX^\intercal-\balpha \end{smallmatrix} \right]$, where $\bG = \left[\begin{smallmatrix}
2\bA & \bE - \bX ^\intercal \\
\bE^{\intercal} -\bX & 2\bC
\end{smallmatrix}\right]$, $\bzero^{(j,k)}$ is a $j\times k$ matrix of zeros,
\begin{equation}
		\bM_1 = \,\tfrac{1}{2}
 \left[ \begin{smallmatrix} \bzero^{(n,n)} & \bzero^{(m,m)} \\ \bzero^{(m,m)} & \ba^{-1} \end{smallmatrix}\right],
\quad
		\bM_2 = \left[ \begin{smallmatrix} -\bkappa & \bzero^{(n,m)} \\
			\bzero^{(m,n)} & \bzero^{(m,m)} \end{smallmatrix} \right],
\quad
\text{and}
\quad
		\bM_3 =  \left[ \begin{smallmatrix} \bzero^{(n,n)} & - \bkappa\, \bX ^\intercal \\
			 - \bX  \, \bkappa^\intercal & -2\phi \,\tilde{\bSigma}  \end{smallmatrix} \right].
%\label{eq:matrix riccati block M3}
\label{eq:matrix riccati block M}
\end{equation}%
%\end{subequations}%
\item The linear matrix PDEs
\begin{subequations}
	\begin{align}
	\dot\bB + \mathcal{L}^\bmu \bB -\bkappa \,\bB +\tfrac{1}{2}\,(\bE^{\intercal}-\bX)^{\intercal}\,\ba^{-1}\,\bD &= \bzero^{(n)}\,, \\
	\dot\bD + \mathcal{L}^\bmu\bD + \bC^{\intercal}\,\ba^{-1}\,\bD + 2\,\phi\,\tilde{\bSigma}\;\bmQ_{t} + \bX\;\overline{\bb}\;\bmu &= \bzero^{(m)}\,,\label{eqn:PDE for D}\\
	\dot F + \mathcal{L}^\bmu F +\tfrac{1}{4}\bD^{\intercal}\,\ba^{-1}\,\bD+Tr(\bSigma\,\bA)-\phi\;\bmQ_{t}^{\intercal}\;\tilde{\bSigma}\;\bmQ_{t} &= 0\,,
	\end{align}
\label{eq:matrix_diff_equation}%
\end{subequations}
with terminal conditions
\[
\bB(T,\cdot) = \bzero^{(n)},  \quad \bD(T,\cdot) = \bzero^{(m)}, \quad
F(T)=\btheta^{\intercal}\,\bX ^{\intercal}\bQ_{0} - \frac{1}{2}\,\bQ_0^\intercal\,\bX\,\bb\,\bX^\intercal\,\bQ_0\,,
\]
and $\bzero^{(k)}$ denotes a vector of $k$ zeros.
\end{enumerate}
In the above, the dot notation denotes time derivative.
 \label{prop:solving_DPE}
\end{proposition}

\begin{proof}
See \ref{sec:proof solve DPE}.
\end{proof}

The following theorem shows that the solution to (\ref{eq:matrix_riccati_block_form}) is bounded on $ [0, T] $,
as long as we choose the terminal penalty $ \balpha $ to be large enough.
\begin{theorem}\label{thm:bounded}
	If $ \frac{1}{2} \,\bX\,\bb\,\bX^\intercal - \balpha $ is negative definite, the matrix Riccati differential equation \eqref{eq:matrix_riccati_block_form} has a bounded solution on $[0, T]$.
\end{theorem}
\begin{proof}
	See \ref{sec:Proof bounded}.
\end{proof}

Furthermore, the linear matrix PDEs \eqref{eq:matrix_diff_equation} admit a unique probabilistic representation.
\begin{theorem} \label{thm:feynman-kac}
	Suppose the assumption of Theorem \ref{thm:bounded} are enforced, and further
	$ \mathbb{E} [\,|\bmu^\pm_0|^2\,] < \infty $ and  there exists a constant $C$ such that %$\mathbb E_{0,\bmu}\left[\; | \bmu_t^{\pm}| \; \right] < C_1(1+|\bmu^{\pm}|)$
	and $\mathbb E_{0,\bmu}\left[\;|\bmu_t^{\pm}|^2\; \right] < C(1+|\bmu^\pm|^2)$ for all $t\in[0, T]$. Let $\bB(t,\bmu) $, $ \bD(t,\bmu) $ and $F(t,\bmu)$ be $ C^{1,2}([0,T), \mathbb{R}^m) $ solutions to (\ref{eq:matrix_diff_equation}),
	each with quadratic growth in $ \bmu $, uniformly in $ t $, then
	\begin{subequations}
		\begin{align}
			\bD(t, \bmu) = &  \int_t^T \boldsymbol{:}\be^{\int_t^u \bC^\intercal(s)\,\ba^{-1}\,ds}\boldsymbol{:}\; \left\{ 2\,\phi\;\tilde{\bSigma} \,\bmQ_u + \;\bX\;\overline{\bb}\;\mathbb{E}_{t,\bmu} \left[\;\bmu_u\;\right] \right\} \; du \, , \label{eq:feynman-kac_D} \\
			\bB(t, \bmu) =&\, \tfrac{1}{2} \int_t^T e^{-\bkappa\,(u - t)} \;\left(\bE^\intercal - \bX\right)^\intercal \;\ba^{-1}
\;\mathbb{E} \left[ \bD(t, \bmu_u) \right] \,du\,, \\
			F(t\,\bmu) =& \int_t^T
			\left\{
			\tfrac{1}{4}\;\mathbb{E}_{t,\bmu}\left[\bD^{\intercal}(u,\bmu_u)\,\ba^{-1}\,\bD(u,\bmu_u)\right]  + Tr(\bSigma\,\bA(u))-\phi\;\bmQ_{u}^{\intercal}\,\tilde{\bSigma}\,\bmQ_{u}
			\right\}du\,,		
		\end{align}		
	\end{subequations}
where the notation $\boldsymbol{:}\be^{\int_u^t \cdot\,ds}\boldsymbol{:}$ represents the time-ordered exponential.\footnote{Recall that the time-ordered exponential of a time dependent matrix $\bA(t)$ is defined as $\boldsymbol{:}\be(\int_u^t \bA(s)\,ds)\boldsymbol{:}=\lim_{||\Pi||\downarrow 0 } \prod_{i=1}^{n_{\Pi}}e^{A(t_{i-1})\,\Delta t_i}$, where $\Pi:=\{u=t_0,t_1,\dots,t_n=t\}$ is a partition of $[u,t]$, and $\Delta t_i=(t_i-t_{i-1})$.}

\end{theorem}

\begin{theorem} \label{thm:verification1}
\textbf{Verification.}
Suppose the assumptions in Theorem \ref{thm:bounded} and Theorem \ref{thm:feynman-kac} are enforced,
then the candidate value function (\ref{eq:ansatz1}) is indeed the solution to the control problem. Moreover, the trading rate given by
	\begin{equation}\label{eq:optimal_control}
		\bnu_{t}^{*} = -\tfrac{1}{2}\,\ba^{-1}\,\Big\{\; 2\,\bC(t) \,\bQ_{t}^{\bnu^*} + \left(\bE^{\intercal}(t) - \bX \right)\,(\bS_t-\btheta) + 		 \bD(t,\bmu_t) \; \Big\} \,,			
	\end{equation}
is admissible and optimal.
\end{theorem}

\begin{proof}
	See \ref{sec:proof_of_verification_theorem}.
\end{proof}

The optimal trading rate can be interpreted as the liquidation strategy of AC, plus modifications due to co-integration, order flow impact, and target inventory. In our setup, we obtain the  AC strategy to liquidate a portfolio of $m$ assets by
removing co-integration (setting $\bkappa = 0$ in (\ref{eqn: cointegrated midprices})),
removing the impact of order flow (setting $\overline{\bb} = 0$ in (\ref{eq:price_dynamics}))
and setting the target inventory schedule $\bmQ_t = \bzero$ for all $t  \in [0,T]$.
Note that when we trade in all assets, so that $m=n$, $\bX$ becomes the identity matrix.
With these assumptions $\bA(t) = \bzero^{(n,n)}$, $\bB(t) = \bzero^{(n)}$, $\bD(t) = \bzero^{(m)}$,  $\bE(t) = \bI_m$ (an $m\times m$ identity matrix), for all $t \in [0, T]$ in \eqref{eq:matrix_riccati_block_form}.
Therefore, the multi-asset  AC strategy is to trade at the speed
	\begin{equation}\label{eq:optimal_control AC}
		\bnu_{t}^{AC} = -\,\ba^{-1} \; \bC(t) \; \bQ_{t}^{\bnu^{AC}} \, .		
	\end{equation}

The difference between the optimal trading strategy \eqref{eq:optimal_control} and the AC
strategy consists of two components. The first is $\left(\bE^{\intercal} - \bX \right)\,(\bS_t-\btheta)$,
which accounts for co-integration in prices, and allows the trader to take advantage of price deviations from all assets -- not just the ones she is trading. This modification vanishes as the strategy approaches the end of the trading horizon because the terminal conditions enforce
$\bE^\intercal\xrightarrow[]{t\to T}\bX$.

The second component, $\bD(t,\bmu)$,   is  the adjustment due to the inventory target and order flow. For example,
if the agent targets zero inventory throughout the life of the strategy,
$\bmQ_t = \bzero$ for all $t$,  and ignores the effect of order flow from other traders, i.e., $\bar{\bb} = \bzero$, then $\bD(t,\bmu)$ vanishes. Moreover, the terminal conditions for $\bD$ imply that the effect of this adjustment in the trading strategy diminishes
as the strategy approaches the terminal date.

As a final point, if the order flow of other agents $\bmu$ is affine, specifically, if $\EE[\bmu_u | \mathcal F_t ] = \balpha(t;u) + \bbeta(t;u)\,\bmu_t$ for some deterministic functions $\balpha(t;u)$ ($dim=n\times1$) and $\bbeta(t;u)$ ($dim=n\times n$), then the PDE system \eqref{eq:matrix_diff_equation} admits the affine ansatz for $\bB(t,\bmu)=\bB_0(t) + \bB_1(t)\,\bmu$ and $\bD(t,\bmu)=\bD_0(t) + \bD_1(t)\,\bmu$, while $F(t,\bmu)= F_0(t) + \bF_1^\intercal(t)\,\bmu + \bmu^\intercal \bF_2(t)\bmu$. Two examples of  affine order flow models are (i) the shot-noise processes, where order flow jumps up at Poisson times and mean-reverts back to zero, with idiosyncratic upward and downward jumps, as well as co-jumping order flow; and (ii) the multivariate Hawkes process, where increases in order flow induces excitation in order flow among a subset of assets or all assets. Both of these models have appeared in a number of papers that study the empirical aspects of order flow.

\subsection{Guaranteed liquidation} \label{sec:limiting_case}

To ensure full liquidation by the end of the trading window, i.e., $\bQ_{T}^{\bnu^*}=\bzero$, we make the liquidation penalty arbitrarily expensive, i.e., all the components of $\balpha$ go to $\infty$. In this  case,  the terminal condition for $\bC$, see \eqref{eq:matrix_riccati_block_form}, becomes arbitrarily large as the entries of the terminal condition $\bC(T)$ go to $ -\infty$.
Now, let us assume that in this limiting case $\bA$, $\bC$, and $\bE$ have the asymptotic series
\begin{equation}  \label{eq:asy_expansion}
	\bA(t)   =   \sum_{n=0}^{\infty} \mathscr{A}_n\, \tau^n\,, \qquad
	\bC(t) = \sum_{n=-1}^{\infty} \mathscr{C}_n \,\tau^n\,, \qquad
	\bE(t) = \sum_{n=0}^{\infty} \mathscr{E}_n \,\tau^n\,,
%\label{eq:asy_power_series}
\end{equation}
where $\tau = T-t$ and $\mathscr A_n$, $\mathscr C_n$ and $\mathscr E_n$
are constant matrices with the same dimensions as $\bA$, $\bC$, and $\bE$ respectively.

The terminal conditions imply that the first term in the asymptotic series $\bA(t)$ is  $\mathscr A_0 = 0$ and in $\bE(t)$ is $\mathscr E_0 = \bX^\intercal$. Moreover,  substituting (\ref{eq:asy_expansion}) into \eqref{eq:matrix_riccati_block_form}, and
matching terms with the same power in $\tau$, we obtain  the following coefficients for the series of $\bA(t)$, $\bC(t)$, and $\bE(t)$
\begin{equation}\label{eq:asy_coefs}
	\mathscr A_1 = 0\,, \quad
	\mathscr C_{-1} = \frac{1}{2} \,\bX\,\bb\,\bX^\intercal -\ba\,, \quad
	\mathscr C_0 = \bzero\,, \quad
	\mathscr E_1 = -\frac{1}{2} \, \bkappa \, \bX^\intercal\,.
\end{equation}
We also show that the term $ \bD(t,\bmu) $ in \eqref{eq:optimal_control} has the asymptotic bound $ \mathcal O(\tau)\,\bmu + \mathcal{O}(\tau) $.
See \ref{sec:detail_asy_calculation} for further details.

Thus, when $\balpha\to\binfty$, so that $\bC(T)\to-\binfty$ (recall that $\bC(T)=\frac{1}{2} \,\bX\,\bb\,\bX^\intercal -\ba$ from \eqref{eq:matrix_riccati_block_form}), we employ the asymptotic series \eqref{eq:asy_expansion} (using the first two terms for each series), to write the optimal liquidation speed \eqref{eq:optimal_control} as follows
\begin{equation}\label{eqn: nu star with infty alpha}
	\bnu_t^*  =  \frac{\bQ_t^{\bnu^*}}{\tau} + \mathcal O(\tau)\, \bZ_t\, + \mathcal O(\tau)\,\bmu + \mathcal{O}(\tau)\,.
\end{equation}
The result is that near maturity, i.e., $\tau\to 0$, the optimal strategy behaves like TWAP (time-weighted-average-price), which is given by the first term in the right-hand side of \eqref{eqn: nu star with infty alpha}. Thus, as the strategy approaches the terminal date, the remaining inventory is liquidated at a constant rate.

\section{Simulations:  Portfolio Liquidation}\label{sec:Numerical-Example}
This section shows  the performance of the strategy  under various assumptions about the set of assets employed  by the investor and illustrates how the strategy performs if the investor does not have enough data to calibrate the model parameters. We first describe how the model parameters are estimated using  exchange data  from five assets traded on the  Nasdaq. Then we compare the performance of the strategy   to  that of AC when the agent liquidates a portfolio of two assets: using the information of another additional set of three stocks, using only the information of the two assets, and allowing asset repurchases. In particular, Subsection \ref{subsubsec: correct parameter estimates} assumes that the investor has enough data to correctly estimate the parameters of the model, and Subsection \ref{subsubsec: incorrect parameter estimates} assumes that the investor does not have access to enough data, so the estimates of the parameters she obtains are incorrect.

 \subsection{Data and model parameters}

 To focus on the additional value added by the co-integration information, we turn off the permanent impact from order flow of other agents, i.e., we set $\overline{\bb}=\bzero$. For an analysis of how agents can benefit from order flow information see \cite{cartea2015incorporating}. We also turn off the permanent impact from the agent's own trading activity, i.e., we assume $ \bb = 0 $.

 We employ high-frequency data from five stocks traded on the Nasdaq exchange: INTC, SMH, FARO, NTAP and ORCL. We use all the messages sent to the exchange in November 3, 2014 to build the LOB at a millisecond frequency. We  sample the best quotes and posted volume  every 60 seconds during the regular trading hours. Midprices are computed as the weighted average of the best bid and the best ask, with weights equal to the volume posted at the best ask and the best bid respectively -- these prices are also referred to as microprice. We remove the first and last half hour to reduce the noise in prices due to the opening and closing auctions. Thus, for the trading day we have a midprice time series  of 330 data points per stock.

The five stocks we employ  are  in the high-tech sector, thus sharing a common trend, so we expect them to be co-integrated. We employ a VAR(1), vector-autoregressive of order 1,  model of the joint midprice dynamics which is a discrete-time version of the price process in  \eqref{eq:price_dynamics}, and apply Johansen's co-integration test to determine the number of co-integrating factors  -- which corresponds to the rank of the matrix $\boldsymbol\kappa$. Table  \ref{table:jcitest_5S60S} reports the p-values  of the co-integration test for the number of co-integrating factors, where $r_i$ corresponds to the null hypothesis that there are at most $i$ co-integrating factors.
\begin{table}[h!]
	\centering
\footnotesize
	\begin{tabular}{cccccc} \hhline{======}
		Model & $r_0$ & $r_1$ & $r_2$ & $r_3$ & $r_4$ \\ \hline
		p-value &  0.001 & 0.223 & 0.409 & 0.4637 & 0.584 \\ \hline
	\end{tabular}
\caption{Johansen's co-integration test for the number of co-integrating factors, and
		$r_i$ corresponds to the null hypothesis that there are at most $i$ co-integrating factors. Nasdaq data November 3, 2014.}
\label{table:jcitest_5S60S}
\end{table}
In Figure \ref{figure:coint_factors}, we show the realisation of this factor through the trading day.

\begin{figure}[b!]
\begin{minipage}{0.49\textwidth}
\begin{center}
\includegraphics[width=0.82\textwidth]{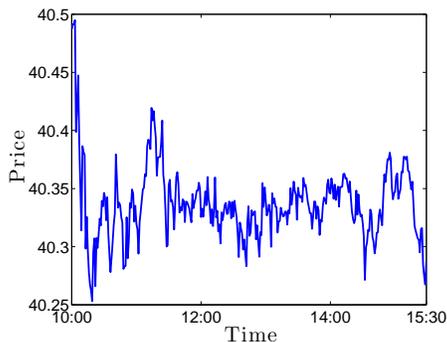}
\end{center}
\end{minipage}
\begin{minipage}{0.49\textwidth}
\caption{In-sample path of the calibrated co-integrating factor over the trading day.}\label{figure:coint_factors}
\vspace{3em}
\end{minipage}
\end{figure}

Here we  choose a midprice  model with 1 co-integration factor, and show in the first two rows in Table \ref{table: theta coint fact a b} the parameter estimates for the mean reverting level $\btheta$ and the co-integration factors of the VAR(1) model (data November 3, 2014). The rest of the table is taken from \cite{cartea2015incorporating} which employs data for the entire year 2014. Rows 3  and 4 are estimates of temporary impact with no cross effects. We choose the temporary impact model with no cross effects to keep our model parsimonious. The bottom 2 rows show the average incoming rates of MOs and their average volume: $\lambda^-$ is the average number of sell MO per hour, $\mathbb E[\eta^-]$ is the average volume of sell MOs. The standard deviation of the estimate is shown in parentheses.

Tables \ref{table:est_kappa_5S60S} and \ref{table:est_Sigma_5S60S} show the mean-reverting matrix $\bkappa$, the variance-covariance matrix $\bSigma$, respectively. Time is $T = 1$, which corresponds to 6.5 hours (1 trading day).

Below we compare the performance of different trading strategies, one of which is AC. In this particular case, we assume that midprices satisfy the SDE
\begin{equation}\label{eq:price_dynamics_AC}	d\bP_t=
(\bsigma^{AC})^{\intercal}\,d\bW^{AC}_{t}\,.
\end{equation}
where $(\bsigma^{AC})^{\intercal}$ is the Cholesky decomposition of the asset prices correlation matrix $\bSigma^{AC} = (\bsigma^{AC})^{\intercal}\,\bsigma^{AC}$. That is, for the AC strategy, asset prices are assumed to be driven  by correlated Brownian motions but are not assumed to be co-integrated.  Table \ref{table:est_Sigma_AC} shows the estimated correlation matrix $\bSigma^{AC}$.

\subsection{Liquidation of portfolio with two assets}\label{subsec: liquidation of portfolio with two assets}
The investor's objective is to liquidate 4,600 shares of INTC and 900 shares of SMH over 1 hour. According to Table \ref{table: theta coint fact a b}, these two numbers correspond to 1\% and 4\%
of the average number of sell volume over 1 hour, respectively. The investor sends MOs at 1 second intervals.

\noindent\textbf{ Strategies.} To illustrate the performance of the liquidation algorithm we consider the following four strategies:
\begin{enumerate}

\setlength\itemsep{0.5em}

	\item \underline{Unrestricted liquidation} (UL): the strategy $\bnu^*_t$ is as in  \eqref{eq:optimal_control}  with  target schedule $\bmQ_t = \bzero$ for all $t$. Recall that the investor's set of admissible strategies does not require the trading speed to remain non-negative. So UL may, if it is optimal, repurchase shares before the end of the trading horizon.

	\item \underline{Restricted liquidation} (RL): the strategy is as that of UL, but the liquidation speed is set to zero if it is optimal to repurchase, i.e., $\max\left(\bnu^*_t,\,0\right)$. This is an ad-hoc adjustment to the optimal strategy to preclude repurchases along the trading window. The max operator $\max(\,\cdot\,, \,0)$ is to be interpreted componentwise: when the speed of liquidation of an asset in the vector  \eqref{eq:optimal_control} is negative, only that component  is set to zero. Finally, the strategy stops trading when inventory hits zero.  The derivation of the  optimal strategy under the constraint $ \bnu^*_t \geq 0 $  is beyond the scope of our study.

	\item \underline{Unrestricted liquidation with target} (ULT): the strategy is as  \eqref{eq:optimal_control} where the  target schedule for each asset is the Almgren-Chriss (AC) strategy, i.e., $\bmQ_t = \bmQ^{AC}_t$ where $\bmQ^{AC}_t$ is the AC liquidation position given by integrating \eqref{eq:optimal_control AC} with penalty parameter $\phi^{AC} = 0.1$.\footnote{The penalty parameter is embedded in  $\bC$ which appears explicitly in the liquidation speed \eqref{eq:optimal_control AC}.} With these parameters, the AC strategy liquidates more than the initial inventory of SMH early on, but repurchases inventory by the trading end. The switching of trading direction is due to the correlation between assets, which causes the trader to take on a hedge-like position to reduce risk.
\item \underline{Almgren-Chriss  liquidation} (AC): the strategy is as in \eqref{eq:optimal_control AC} and the price process is \eqref{eq:price_dynamics_AC},  so the strategy only uses information from INTC and SMH without a co-integrating factor. This is the benchmark we employ to compare the results of the previous three strategies.

\end{enumerate}

\noindent\textbf{ Scenarios.} We simulate $10^6$ sets of sample price paths and look at the performance of the four strategies when liquidating shares in INTC and SMH for a range of values of the penalty parameter $\phi=10^{-2}\times \left\{  0.50\,,   0.54\,,    0.5833\,,    0.63\,,    0.6804\,,    0.7349\,,    0.7937\,,    0.8572\,,    0.9259\,,    1\right\}$, and the liquidation penalty is $\alpha=10^6$ for both assets, i.e., employ strategies that guarantee full liquidation.  We measure the performance by comparing the terminal wealth of UL, RL, ULT, and AC under two scenarios:
\begin{itemize}
  \item Scenario 1. Liquidate shares in  INTC and SMH and  employ the additional information provided by three additional assets: FARO, NTAP, ORCL.

  \item Scenario 2. Liquidate shares in  INTC and SMH and only employ the information provided by the dynamics of INTC and SMH.

\end{itemize}

\begin{figure}[!t]
\begin{center}
%\begin{minipage}{.5\textwidth}
%  \centering
\subfigure[full information]
{\includegraphics[width = 0.4\textwidth]{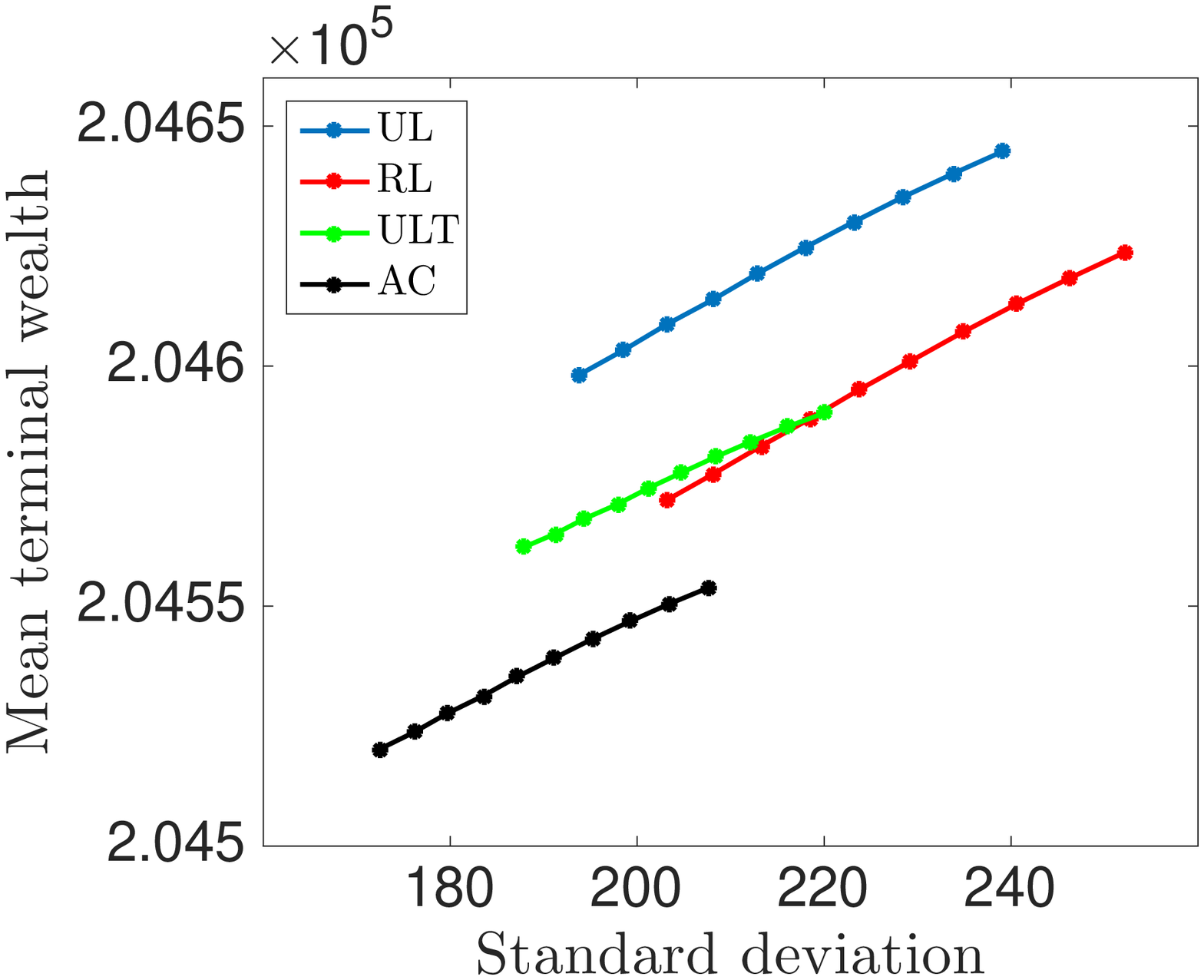}}
%\end{minipage}%
%\begin{minipage}{.5\textwidth}
%\centering
\subfigure[partial information]
{
\includegraphics[width=0.4\textwidth]{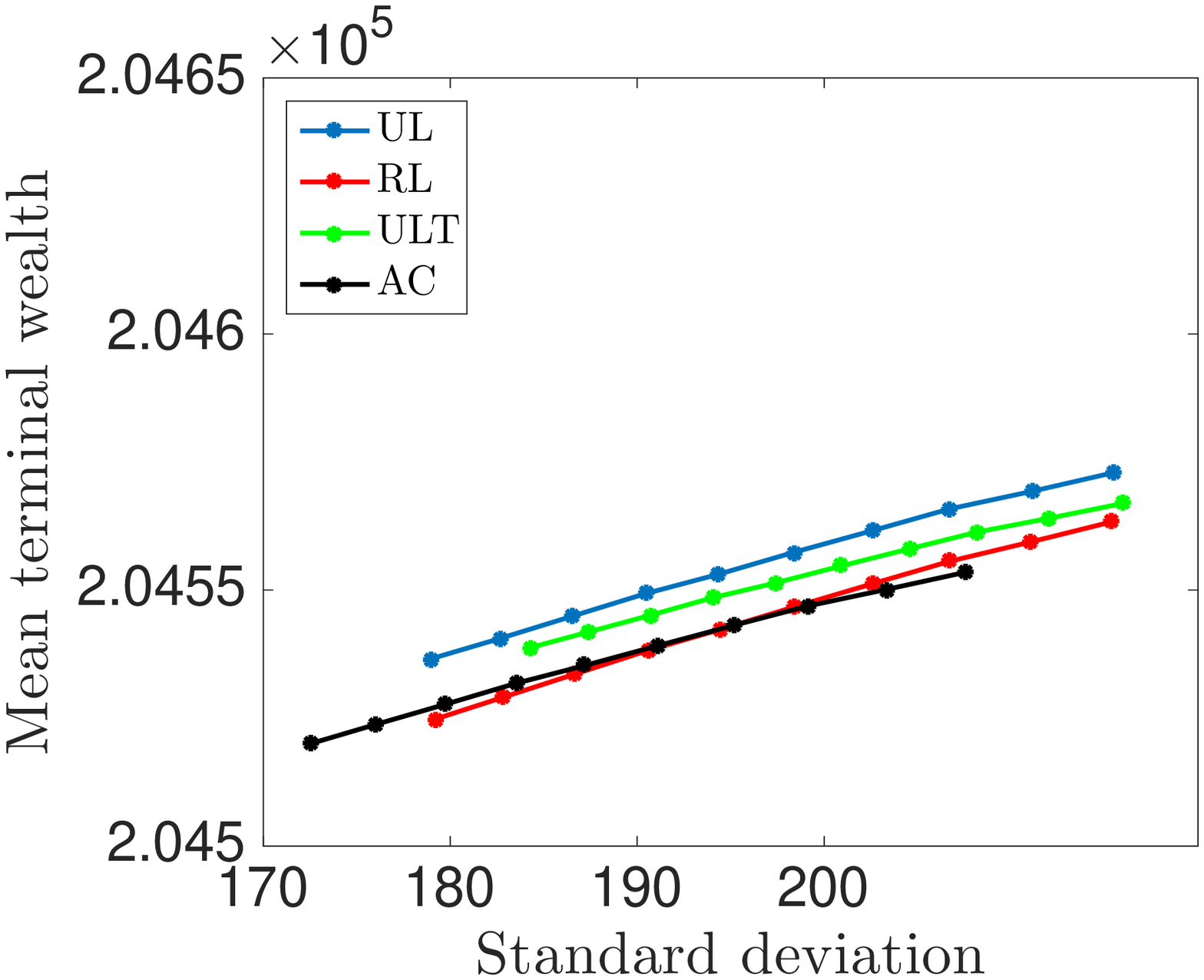}
}
%\end{minipage}
\end{center}
\vspace{-1.5em}
\caption{Trading INTC and SMH.	Risk-reward  for  UL, RL, ULT, and AC. Left (right) panel, strategies employ information from all (only the traded) stocks. Within each panel, the penalty $\phi$ increases moving from the right to the left of the diagrams. \label{figure:risk_reward_all}}
\end{figure}

\subsubsection{Investor estimates model parameters without error}\label{subsubsec: correct parameter estimates}

Figure \ref{figure:risk_reward_all} shows the mean terminal wealth (aggregate cash from liquidating shares in both INTC and SMH) of the four strategies as a function of its standard deviation. As the penalty parameter $\phi $ increases, the standard deviation and mean of the terminal wealth decrease. To see the intuition behind this relationship  let us focus on UL. The agent targets an inventory of zero throughout the life of the strategy, and the value of $\phi$ determines how closely  the strategy tracks this target. When the  penalty is high, the strategy is less able to trade strategically by either speculating (repurchasing shares) and/or taking advantage of midprice signals that stem from the co-integrating factor. Thus, potential benefits from taking advantage of price movements are outweighed by the requirement that inventory must be drawn to zero very quickly. Conversely, as  the penalty  becomes  smaller, the strategy will have more opportunities to anticipate and take advantage of midprice movements and these will not be curbed by a strict inventory target.

The left-hand panel of the figure shows Scenario 1 where  UL, RL, and ULT employ the additional information  provided by FARO, NTAP, ORCL. Clearly, UL dominates the other strategies where AC is the worst performer because it does not account for the co-integration of assets. The right-hand panel of Figure \ref{figure:risk_reward_all} shows Scenario 2 where only information of the co-integrated pair INTC and SMH is employed. Clearly,  not employing the additional information provided by other assets that are co-integrated with those in the liquidating portfolio has a considerable effect on the strategies' performance.

\begin{figure}[t!]

\begin{center}
%\begin{minipage}{0.82\textwidth}
\subfigure[INTC: full information]
{\includegraphics[width=0.4\textwidth]{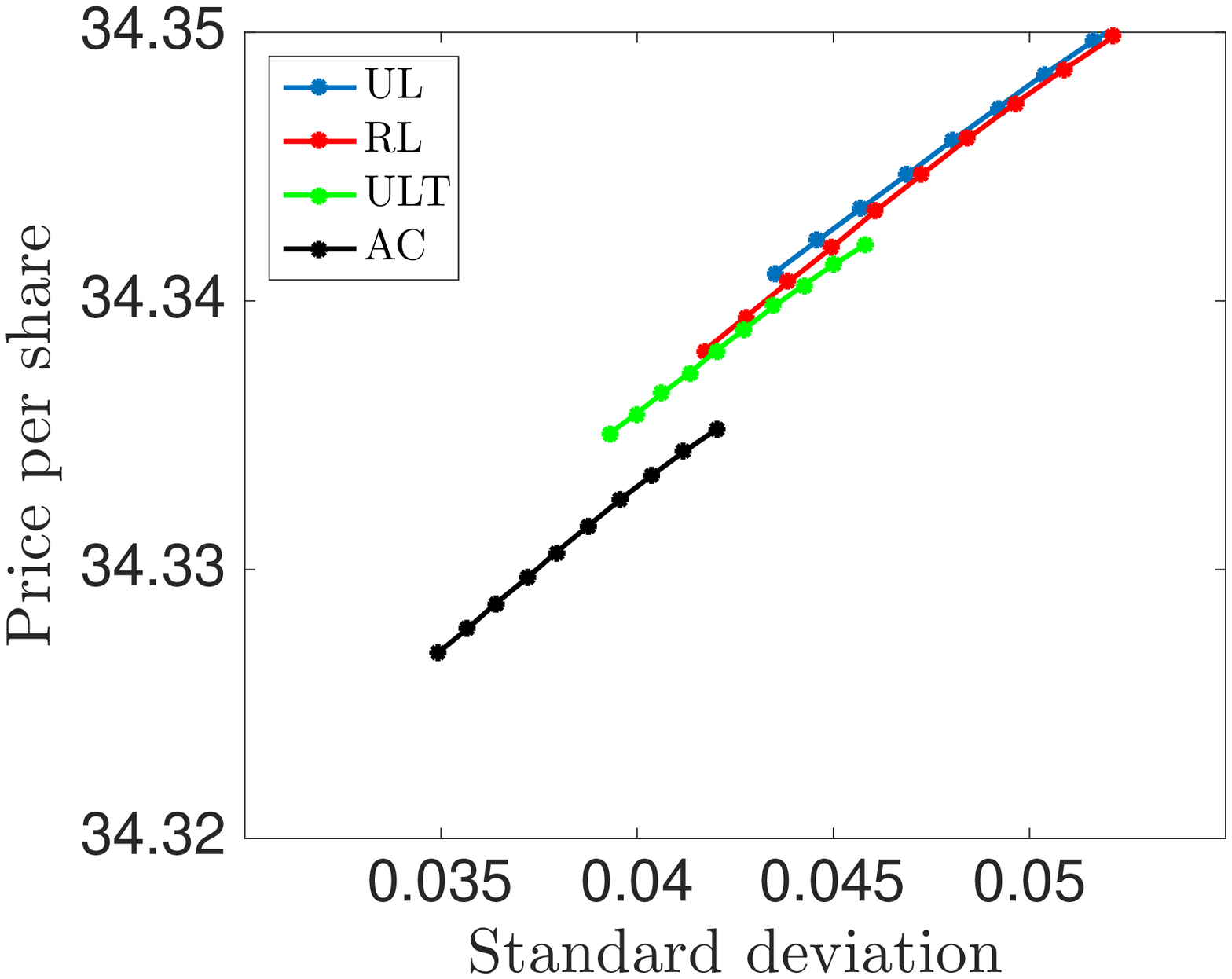}}
\subfigure[INTC: partial information]
{\includegraphics[width=0.4\textwidth]{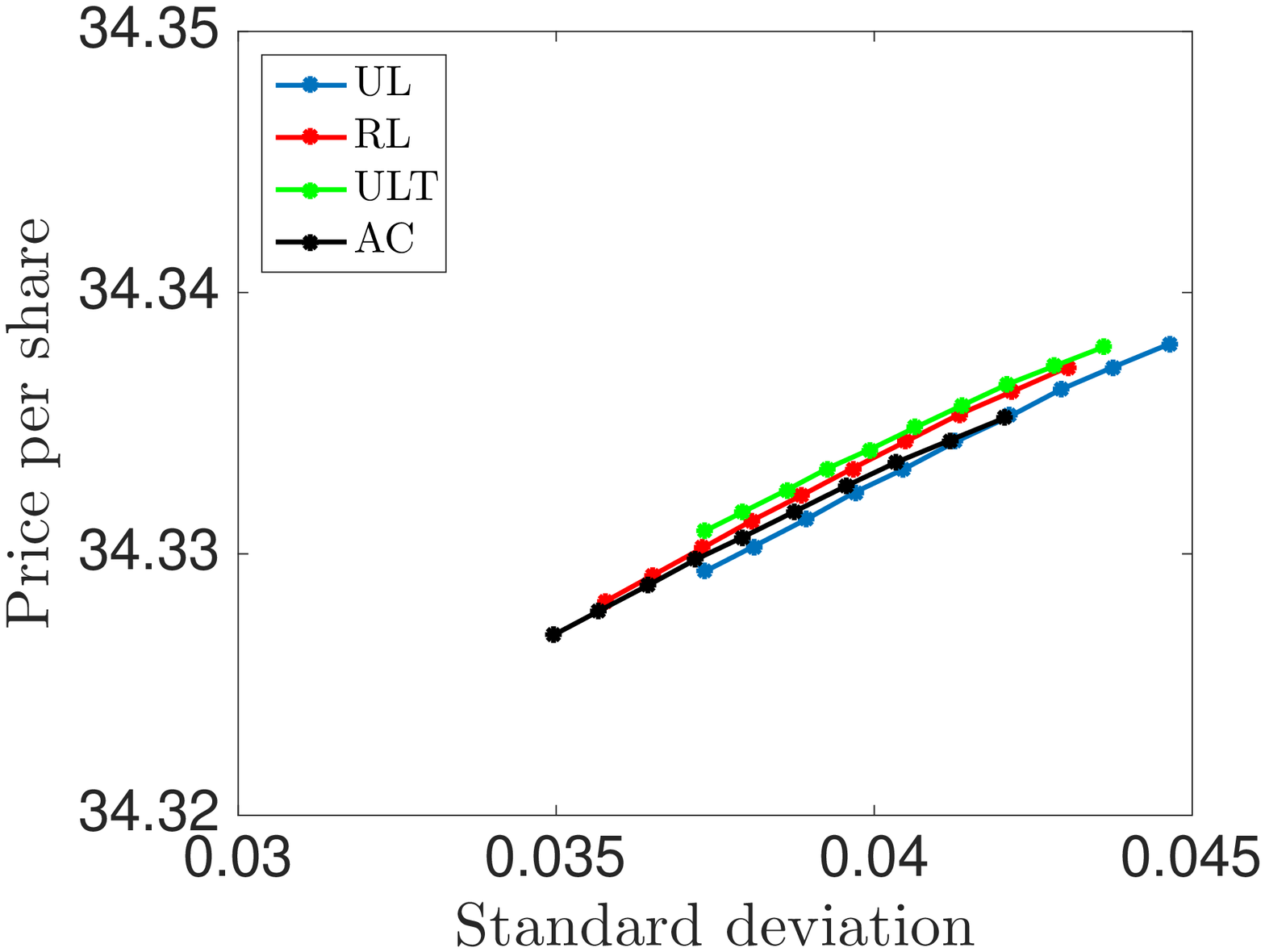}}
\subfigure[SMH: full information]{	\includegraphics[width=0.4\textwidth]{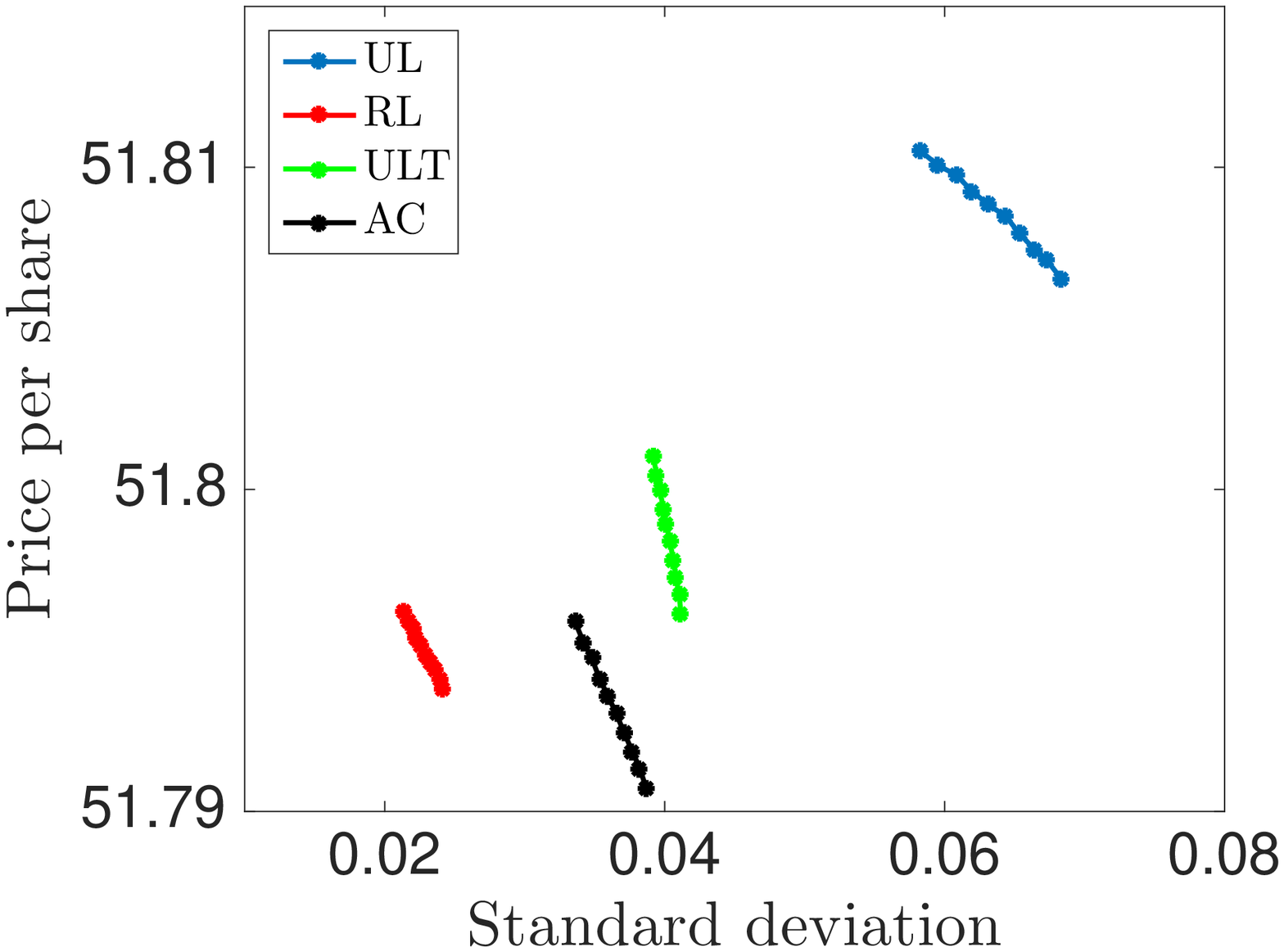}}
\subfigure[SMH: partial information]
{\includegraphics[width=0.4\textwidth]{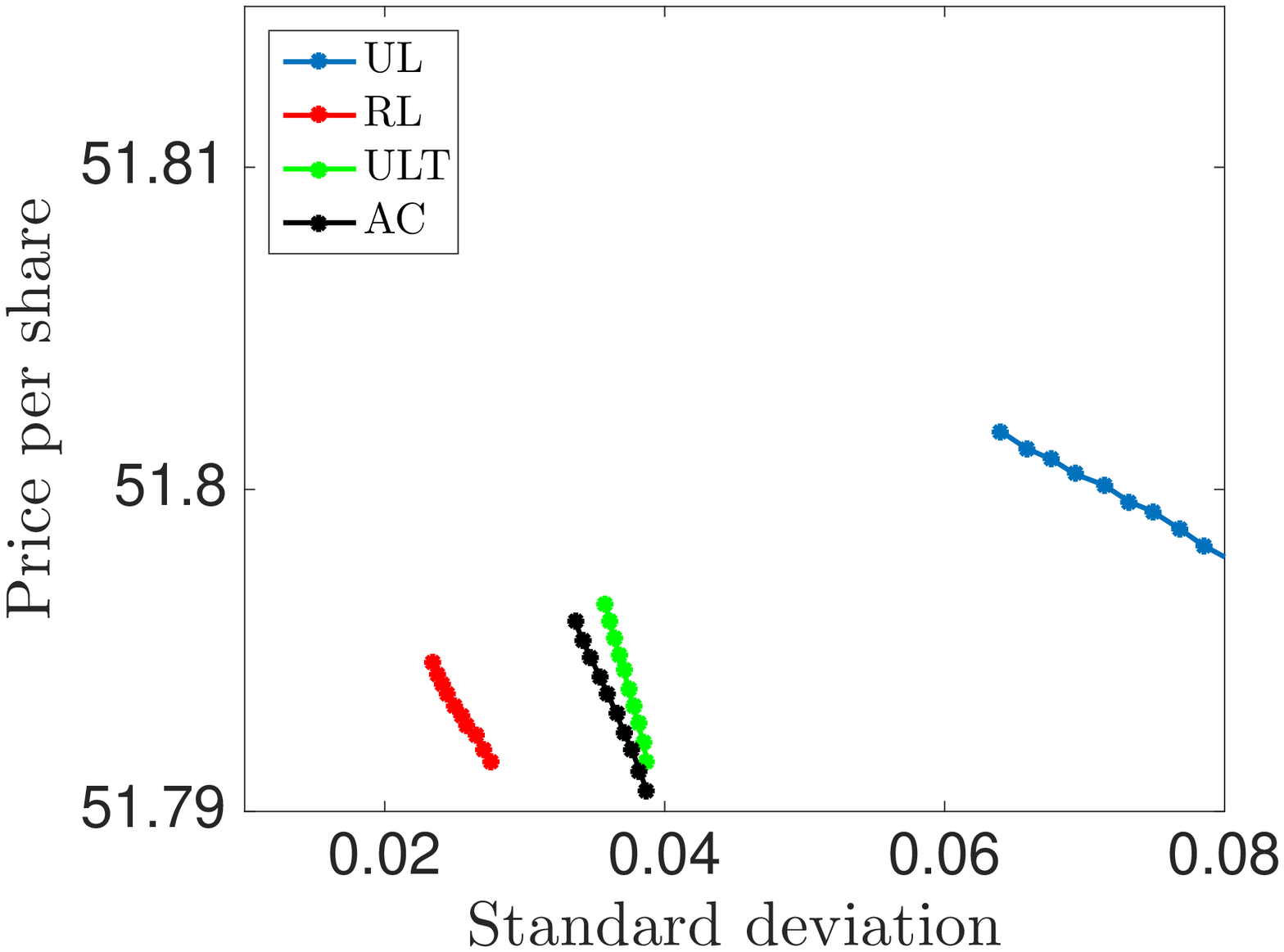}}
\end{center}
\vspace{-1.5em}
\caption{Price per share of INTC and SMH for UL, RL, ULT, and AC. Left (right) panels, strategies employ information from all (only the traded) stocks. Within each panel, the penalty $\phi$ increases moving from the right to the left of the diagrams.\label{figure:risk_reward INTC_SMH}}
\end{figure}
Figure  \ref{figure:risk_reward INTC_SMH}  shows the mean price per share for INTC and SMH, respectively, for a range   of values of the parameter $\phi$. For both shares, the figures in the left-hand panels show that including information from other co-integrated assets boosts the performance of UL. The right-hand panel shows that UL is more volatile than the other strategies and this is a result of the strategy speculating on repurchases of the assets.

Table \ref{tab:summ stats UL RL} shows, for different values of the target penalty parameter $\phi$, how often UL repurchases shares and the percentage of times that UL and RL underperform AC. Here  UL and RL employ information of the midprice dynamics of the five co-integrated assets.  The table   shows that UL's speculative component ranges from 13\% to 18\% in INTC and 63\% to 65\% in SMH.  UL's speculative trades are similar to those employed in pairs trading, which take advantage of temporary deviations of prices. Moreover, we observe that very seldom do we see UL underperform AC, whereas RL underperforms in around 13\% to 14\% of the runs. Recall, however, that the optimal strategy  we derived is the UL strategy, while RL is an ad-hoc sub-optimal adjustment that precludes asset repurchases.
\begin{table}[h!]
\centering \footnotesize
\begin{tabular}{l|l|rrr|rrr} \hhline{========}
\multicolumn{2}{c|}{Strategy}                             & \multicolumn{3}{c|}{UL} & \multicolumn{3}{c}{RL} \\ \hline
\multicolumn{2}{c|}{$\phi$}                                 & 1E-2    & 7.3E-3  & 5E-3  & 1E-2    & 7.5E-3  & 5E-3  \\ \hline
\multicolumn{2}{l|}{$\% \nu_{INTC} < 0$}                  & 18.7   & 16.3     &  13.0      & 0     & 0     & 0      \\
\multicolumn{2}{l|}{$\%\nu_{SMH} < 0$}                     & 64.8   & 64.5  &  63.3  & 0    & 0     & 0      \\
%\multicolumn{2}{l|}{$\%\nu_{INTC} < 0$ or $\%\nu_{SMH} < 0$} & 75.3   & 72.0  & 67.0   & 0     & 0     & 0      \\
\multicolumn{2}{l|}{$\% X_T < X_T^{AC}$}                    & 0.2    & 0.4   & 0.9    & 14.2  & 13.9  & 13.4     \\
\hhline{========}
\end{tabular}
\caption{Repurchase frequency for UL, and underperformance of UL and RL with respect to AC.}\label{tab:summ stats UL RL}
\end{table}

Furthermore, as the value of the parameter $\phi$ decreases, there are fewer  instances in which the liquidation speeds for INTC and SMH are negative. At first this might seem counterintuitive, for one expects a more relaxed penalty parameter to allow UL more freedom to speculate. Note however, that a high value of $\phi$ (recall that for UL the inventory-target is $\bmQ_t = \bzero$ for all $t$)  pushes the inventory close to zero early. And once the inventory in both assets is low, the strategy  attempts more speculative trades by repurchasing the asset. These speculative trades are small in  volume, but frequent.

Figure \ref{figure:savings_all} compares the performance of UL and RL with that of AC. The comparison is in basis points according to the commonly used metric
\begin{equation}\label{eqn: bps}
	\text{Savings}^{j} = \frac{X^{j}_{T} - X^{AC}_T}{X^{AC}_{T}}\times 10^4\,,
\end{equation}
where $X^{j}_{T}$ is the terminal cash\footnote{Recall that we have chosen a very large terminal penalty, so that inventory paths end at zero, and hence the terminal cash the agent has equals her wealth from liquidating the shares.} obtained from liquidating the two-asset portfolio employing strategy  $j \in \{\text{UL},\, \text{RL}\}$. In the left-hand panel the strategy employs information from the price dynamics of the five co-integrated assets. For UL, savings are in the order of 4 to 4.5 basis points, and for RL between 2.5 and 3.5 basis points. In the   right-hand panel,  only information provided by the midprice dynamics of the two-asset portfolio is employed, so  as expected,  the savings are lower.
\begin{figure}[t!]
\begin{center}
%\begin{minipage}{.5\textwidth}
%  \centering
\subfigure[full information]
{
\includegraphics[width=0.4\textwidth]{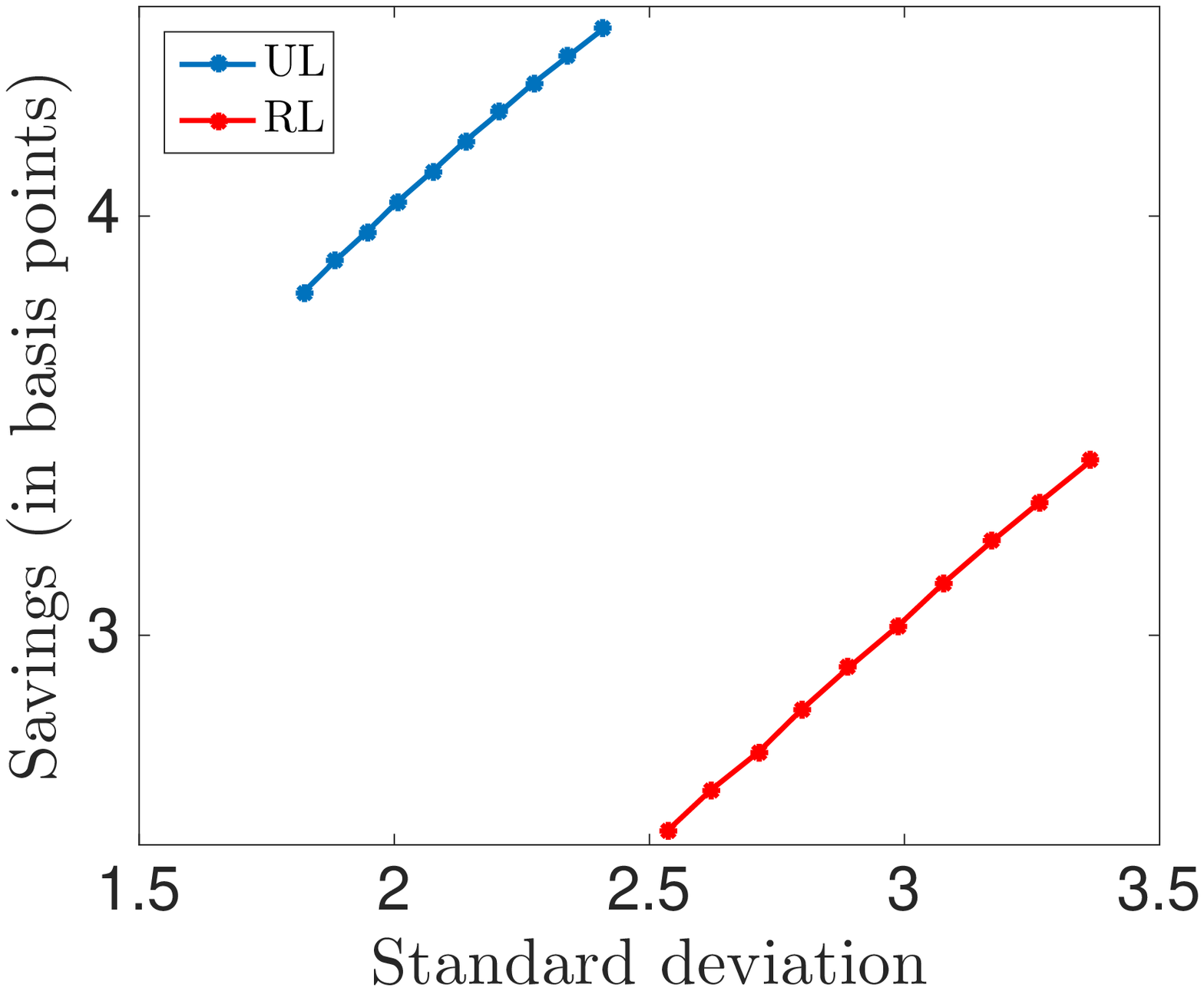}
}
%\end{minipage}%
%\begin{minipage}{.5\textwidth}
%\centering
\subfigure[partial information]
{
	\includegraphics[width=0.4\textwidth]{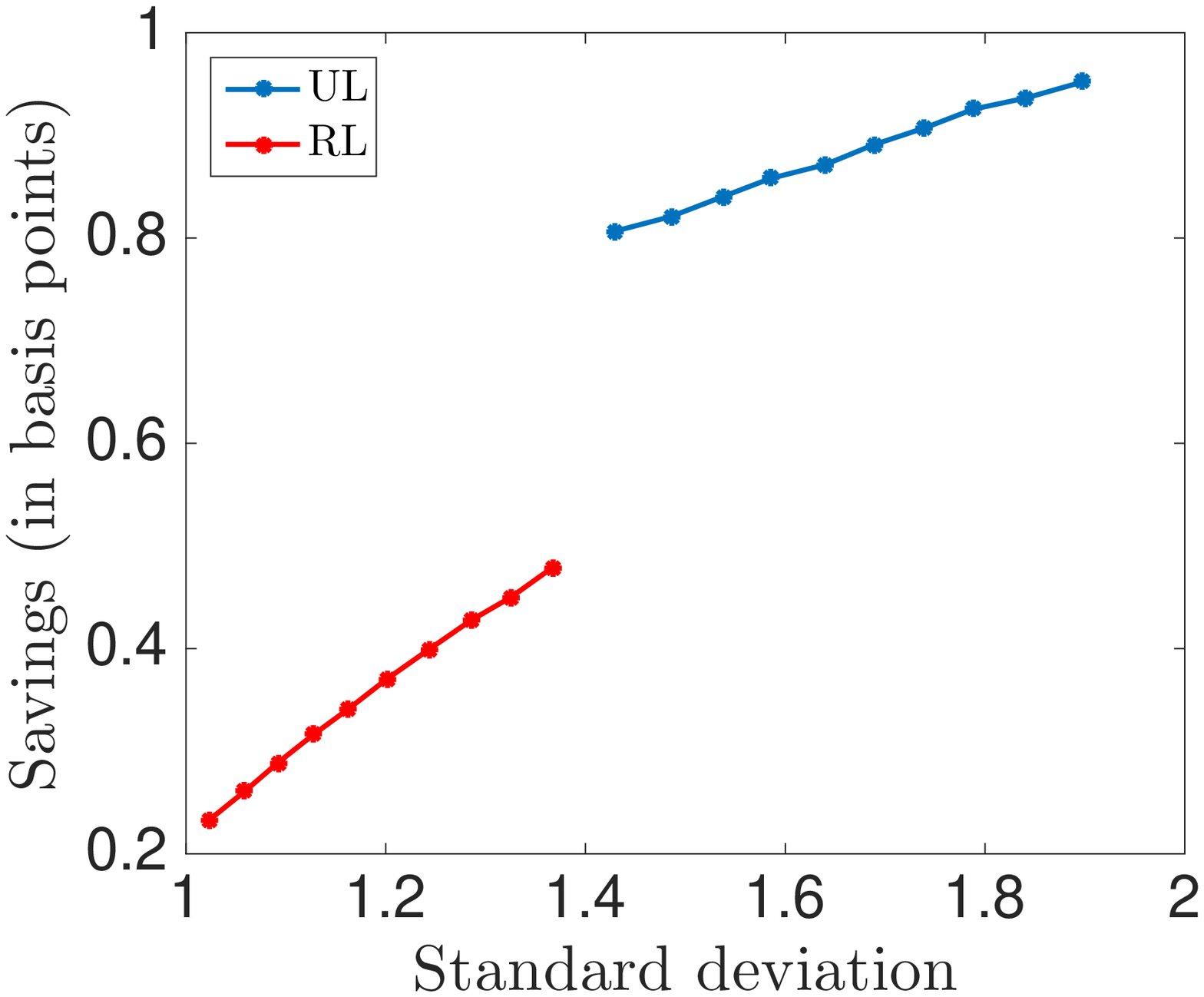}
}
%\end{minipage}
\end{center}
\vspace{-1.5em}
\caption{UL and RL savings in basis points  compared to AC. Left (right) panel, strategies employ information from all (only the traded) stocks. Within each panel, the penalty $\phi$ increases moving from the right to the left of the diagrams. \label{figure:savings_all}}
\end{figure}

Finally, Table \ref{tab: basis points} shows the quantiles of performance of UL and RL measured using \eqref{eqn: bps} for a range of the penalty parameter $\phi$.  The strategies use the information of the five co-integrated assets.
\begin{table}[h!]
	\centering \footnotesize
	\begin{tabular}{l|l|rrr|rrr} \hhline{========}
		\multicolumn{2}{c|}{Strategy}     & \multicolumn{3}{c|}{UL} & \multicolumn{3}{c}{RL} \\ \hline
		\multicolumn{2}{c|}{$\phi$}       & 1E-3    & 7.5E-4  & 5E-4  & 1E-3    & 7.5E-4  & 5E-4 \\ \hline
		\multirow{5}{*}{\rotatebox[origin=c]{90}{quantile}}
		& 5\%  & 1.15 & 1.25 & 1.27 & -1.54 & -1.34 & -1.21 \\
		& 25\%  & 2.77 & 2.60 & 2.55 & 1.16 & 0.96 & 0.84 \\
		& 50\%  & 4.11 & 3.81 & 3.60 & 3.05 & 2.64 & 2.29 \\
		& 75\%  & 5.73 & 5.25 & 4.85 & 5.28 & 4.54 & 4.01 \\
		& 95\%  & 8.86 & 7.80 & 7.18 & 9.57 & 7.94 & 7.13 \\
		\hhline{========}
	\end{tabular}
	\caption{Quantiles of relative savings, measured in basis points using \eqref{eqn: bps}.}\label{tab: basis points}
\end{table}

\subsubsection{Investor estimates model parameters with error}\label{subsubsec: incorrect parameter estimates}

Here we assume that the investor does not have access to enough trading data, so the parameter estimates she obtains are incorrect. The investor observes prices for one day for each asset, which she employs to calibrate the model. The prices she observes are simulated using the parameters in Tables  \ref{table:est_kappa_5S60S} and   \ref{table:est_Sigma_5S60S}. From the observed data, the investor samples  prices every minute to estimate parameters, which are reported in Tables   \ref{table:est_kappa_5S60S_with_error} and  \ref{table:est_Sigma_5S60S_with_error}. Moreover, 	 we  use the same set of prices to estimate the coefficients for the benchmark AC strategy -- parameter estimates are reported in Table \ref{table:est_Sigma_AC_with_error}.

To illustrate how the strategy performs when the model parameters are incorrect, we first proceed as above. Then, we simulate $10^6$ sets of sample price paths (using the original parameters -- so that the agent has incorrect parameters in their trading strategy) and look at the performance of the four strategies when liquidating shares in INTC and SMH, and proceed as in Subsection \ref{subsec: liquidation of portfolio with two assets}. The results are broadly the same as those obtained in the previous section when the investor's parameter estimates were the same as those used to simulate price paths. For example, Table \ref{tab: basis points_with_error}   shows quantiles of relative savings, measured in basis points using \eqref{eqn: bps}, when parameters are estimated \textbf{with error}. The results in the table are similar to those show above in Table \ref{tab: basis points} when the investor estimated the parameters of the model without error. Finally, we do not present the analogues to the figures shown above because the results are qualitatively the same.

\begin{table}[h!]
	\centering \footnotesize
	\begin{tabular}{l|l|rrr|rrr} \hhline{========}
		\multicolumn{2}{c|}{Strategy}     & \multicolumn{3}{c|}{UL} & \multicolumn{3}{c}{RL} \\ \hline
		\multicolumn{2}{c|}{$\phi$}       & 1E-2    & 7.5E-3  & 5E-3  & 1E-2    & 7.5E-3  & 5E-3 \\ \hline
		\multirow{5}{*}{\rotatebox[origin=c]{90}{quantile}}
		& 5\%  & 0.47 & 0.77 & 1.05 & -2.12 & -1.90 & -1.72 \\
		& 25\%  & 2.90 & 2.82 & 2.80 & 1.23 & 0.97 & 0.87 \\
		& 50\%  & 4.74 & 4.45 & 4.21 & 3.53 & 3.05 & 2.71 \\
		& 75\%  & 6.92 & 6.40 & 5.92 & 6.20 & 5.38 & 4.73 \\
		& 95\%  & 11.06 & 9.91 & 8.94 & 11.21 & 9.58 & 8.28 \\
		\hhline{========}
	\end{tabular}
	\caption{Quantiles of relative savings, measured in basis points using \eqref{eqn: bps}, when parameters are estimated \textbf{with error}.}
	\label{tab: basis points_with_error}
\end{table}

%% Section: Conclusions
\section{Conclusions}\label{sec: concs}

We show how to liquidate a basket of assets whose prices are co-integrated. In our framework, market orders from all participants, including the agent liquidating the basket, have a permanent impact  on asset prices. In addition, the agent receives prices that are worse than the best quotes because her trades  walk the limit order book, i.e., have temporary price impact.  We assume that price impact is linear in the speeds of trading and order flow has  cross-effects: trade activity in one asset may have a permanent effect on prices of co-integrated assets and a temporary effect on the limit order books that display the liquidity of the co-integrated assets.

The agent maximizes terminal wealth and targets an inventory schedule. The liquidation strategy employs information from $n$ co-integrated assets and liquidates a basket consisting of a subset of  $m\leq n$ assets. We estimate the model parameters and co-integration factors using trade data from five stocks (INTC, SMH, FARO, NTAP, and ORCL) in the Nasdaq exchange. The agent's basket consists of 4,600 shares in INTC and 900 shares in SMH. We compare the performance of the strategy, under various assumptions,  to that of AC where  the agent models the correlation between the assets in the basket, but  does not model co-integration or employ additional information from other assets.

Our simulations of the liquidation program  show that additional information from other co-integrated stocks considerably boosts the performance of the strategy. For example,  if the level of urgency required by the agent to liquidate the portfolio is high (resp. low)  the strategy outperforms  AC by 4 (resp. 4.5) basis points.  This improvement over AC is due to the quality of the information provided by the co-integrated assets, and  due to a speculative component of the strategy which allows the agent to repurchase shares during the liquidation horizon to take advantage of price signals. If the agent is not allowed to speculate, i.e., cannot repurchase shares, the relative savings compared to AC, depending on the level of urgency, are between 2.5 to 3.5 basis points.

\clearpage

\appendix

\section{Proofs}

\subsection{Proof of Proposition \ref{prop:solving_DPE}}\label{sec:proof solve DPE}
\begin{proof}
	Substituting  the ansatz (\ref{eq:ansatz1}) into (\ref{eq:HJB_equation}),
	we see that $\mathcal{L^\nu}H$ can be simplified to
	\begin{equation}
\begin{split}
\mathcal{L}^{\nu}H =&	\bnu^\intercal \,\ba\, \bnu - \bnu^\intercal \,\Big( \left(\bE^\intercal - \bX \right)\,\bz + 2\,\bC\,\bbq + \bD \Big)
\\
& + \bmu^\intercal\,\bb^\intercal\,\bX^\intercal\,\bbq -\bz^{\intercal}\left(\bkappa\,\bA + \bA\,\bkappa^\intercal \right) \bz
 		- \bz^\intercal\,\bkappa \left(\bE\,\bbq+\bB\,\right)+\text{Tr}\left(\bSigma^{u}\bA\right)\,.
 \end{split}
\label{eqn:all_terms_with_nu}
\end{equation}
The supremum of (\ref{eqn:all_terms_with_nu}) is achieved at
\begin{eqnarray}
		\bnu^{*} & = & -\tfrac{1}{2}\ba^{-1} \left( 2\,\bC\,\bbq+
			\left(\bE^{\intercal} - \bX \right)\,\bz + \bD \right)\,.			
		\label{eq:optimal_trading_rate}
\end{eqnarray}
Substituting $\bnu^{*}$ into (\ref{eq:HJB_equation}) we obtain the following equality
\begin{eqnarray*}
 0 & = & \bz^\intercal\, \dot{\bA}\, \bz + \bz^\intercal\,\left(\dot\bB + \mathcal{L}^\bmu\bB\right) + \bbq^\intercal \,\dot{\bC} \,\bbq + \bbq^\intercal \,\left(\dot\bD + \mathcal{L}^\bmu\,\bD\right) + \bz^\intercal \dot{\bE} \,\bbq + \dot{F} + \mathcal{L}^\bmu F \\
 & & - \phi \,\left(\bbq - \bmQ_t\right)^\intercal \,\tilde{\bSigma} \,\left(\bbq - \bmQ_t\right) -\bz^{\intercal}\left(\bkappa\,\bA + \bA\,\bkappa^\intercal \right) \bz
 - \bz^\intercal\,\bkappa \left(\bE\,\bbq+\bB\,\right) + \text{Tr}\left(\bSigma^{u}\,\bA\right) \\
  & &+ \frac{1}{4} \left( 2\,\bC\,\bbq+
  \left(\bE^{\intercal} - \bX \right)\,\bz + \bD \right)^\intercal \,\ba^{-1}\, \left( 2\,\bC\,\bbq+
  \left(\bE^{\intercal} - \bX \right)\,\bz + \bD \right) \,.
\end{eqnarray*}
Matching the coefficients for $\bz^\intercal (\cdot) \bz$, $(\cdot)^\intercal \bz$ $\bbq^\intercal (\cdot) \bbq$, $(\cdot)^\intercal \bbq$, $\bz^\intercal (\cdot) \bbq$ and the constant, and stacking $\bA$, $\bC$ and $\bE$ we obtain the system of matrix Riccati equations  \eqref{eq:matrix_riccati_block_form} and the linear PDEs (\ref{eq:matrix_diff_equation}).
\qed
\end{proof}

\subsection{Proof of Theorem \ref{thm:bounded}}\label{sec:Proof bounded}

In this subsection we show that the solution to matrix Riccati equation (\ref{eq:matrix_riccati_block_form}) remains bounded on $[0,T]$. To show this, we require two intermediate results.

We first state the following comparison theorem (for a proof, see Theorem 2.2.2 in \cite{kratz2011optimal}).
\begin{theorem} \label{thm:comparison_riccati_equation}
	Let $\bmL_1(t)$, $\bmL_2(t)$, $\bmM(t)$, $\bmN_1(t)$, $\bmN_2(t)$ $\in \mathbb{R}^{d \times d}$ be piecewise continuous on $\mathbb R$.
	Moreover, suppose $\bmL_1(t)$, $\bmL_2(t)$, $\bmN_1(t)$, $\bmN_2(t)$ ($t \in \mathbb{R}$) and $\bmS_1$, $\bmS_2$ $\in \mathbb{R}^{d\times d}$
	are symmetric. Let $T > 0$ and
	\begin{equation*}
	\bmS_1 \geq \bmS_2, \qquad \bmL_1 \geq \bmL_2 \geq \bzero, \qquad  \bmN_1 \geq \bmN_2\,,
	\end{equation*}
	on $[0, T]$. Assume that the terminal value problem
	\begin{equation*}
	\dot{\bH_1} + \bH_1 \bmL_1 \bH_1 + \bmM \bH_1 + \bH_1 \bmM + \bmN_1 =  0\,, \qquad \bH_1(T) = \bmS_1\,,
	\end{equation*}
	has a solution $\bH_1$ on $[0, T]$. Then the terminal value problem
	\begin{equation*}
	\dot{\bH_2} + \bH_2 \bmL_2 \bH_2 + \bmM \bH_2 + \bH_2 \bmM + \bmN_2 =  0\,, \qquad \bH_2(T) = \bmS_2\,,
	\end{equation*}
	has a solution $\bH_2$ on $[0, T]$ and $\bH_1(t) \geq \bH_2(t)$ for all $t \in [0, T]$.
\end{theorem}

From the theorem above, we can show the existence of solution to (\ref{eq:matrix_riccati_block_form}) by
bounding it by another matrix Riccati differential equation, for which the solution is bounded.
The candidate we consider is
\begin{equation}
\dot{\bH} + \bH \bM_1 \bH + \bH \bM_2 + \bM_2^\intercal \bH + \tilde{\bM_3} = \bzero\,,
\label{eq:bound_for_matrix_riccati}
\end{equation}
with terminal condition
\[
\bH(T) = \left[\begin{smallmatrix} \bzero & \bzero \\ \bzero & \bX\bb\bX^\intercal -2 \balpha \end{smallmatrix}\right],
\]
where $\bM_1$ and $\bM_2$ are given by (\ref{eq:matrix riccati block M}),
\[
\tilde{\bM_3} =    \left[\begin{smallmatrix} \gamma^{max} \bI_n & \bzero \\
\bzero & \bzero \end{smallmatrix}\right],
\]
and $\gamma^{max}$ is the largest eigenvalue of the matrix $\frac{1}{2\phi} \,\bkappa \,\bX^\intercal \,\tilde{\bSigma}^{-1} \,\bX \,\bkappa^\intercal$.

The following theorem explicitly characterize the solution of (\ref{eq:bound_for_matrix_riccati}).
\begin{theorem}
	\label{thm:solution_riccati_bound}
	Suppose $ \balpha - \frac{1}{2} \,\bX\,\bb\,\bX^\intercal $ is positive definite, the matrix Riccati differential equation (\ref{eq:bound_for_matrix_riccati}) admits the solution:
	\[
	\bH =   \left[ \begin{smallmatrix} \bH^{11} & \bzero \\ \bzero & \bH^{22} \end{smallmatrix}\right],
	\]
	where $\bH^{11}$ is given by
	\begin{equation}
	\bH^{11} (t) = \gamma^{max}\int_t^T e^{\bkappa(t-u)} e^{\bkappa^\intercal(t-u)} du \,,
	\label{eq:solution_H_11}
	\end{equation}
	and $\bH^{22}$ is given by
	\begin{equation}
	\bH^{22} (t) = -((T - t)\ba^{-1} + (2\balpha - \bX\bb\bX^\intercal)^{-1}) ^{-1}.
	\label{eq:solution_H_22}
	\end{equation}	
\end{theorem}

\begin{proof}
	
	First, write $\bH$ in block form: $
	\bH(t) =  \left[\begin{smallmatrix} \bH^{11}(t) & \bH^{12}(t) \\ \bH^{21}(t) & \bH^{22}(t) \end{smallmatrix}\right] .
	$
	From (\ref{eq:bound_for_matrix_riccati}), it is clear that $\bH^{12} = (\bH^{21})^\intercal$.
	Moreover, $\bH^{11}$, $\bH^{12}$ and $\bH^{22}$ satisfy
	\begin{subequations}
		\begin{align}
		\dot{\bH^{11}} + \tfrac{1}{2} \bH^{12} \ba^{-1} \bH^{21} - \bkappa \bH^{11} - \bH^{11} \bkappa^\intercal + \gamma^{max} \bI_n & = \bzero \,, \label{eq:H_11} \\
		\dot{\bH^{12}} + \tfrac{1}{2} \bH^{12} \ba^{-1} \bH^{22} & = \bzero\,, \label{eq:H_12} \\
		\dot{\bH^{22}} + \tfrac{1}{2} \bH^{22} \ba^{-1} \bH^{22} & = \bzero\,. \label{eq:H_22}
		\end{align}
		% \label{eq:matrix_system_equation}
	\end{subequations}
	It is straightforward to verify that (\ref{eq:solution_H_22}) is a solution to (\ref{eq:H_22}).
	From (\ref{eq:H_12}) and the terminal condition, we have $\bH^{12} (t) = 0$ for all $t \leq T$.
	Moreover (\ref{eq:H_11}) becomes
	\[
	\dot{\bH^{11}} - \bkappa \bH^{11} - \bH^{11} \bkappa^\intercal + \gamma^{max} \bI_n = \bzero\,,
	\]
	whose solution is given by (\ref{eq:solution_H_11}).
	\qed
\end{proof}

We now state the proof of \textit{Theorem \ref{thm:bounded}}.
\begin{proof}
	
	Theorem \ref{thm:solution_riccati_bound} asserts that (\ref{eq:bound_for_matrix_riccati}) has a bounded solution on $[0,T]$,
	therefore, by applying Theorem \ref{thm:comparison_riccati_equation}, it suffices to show that $\tilde{\bM_3} \geq \bM_3$.
	
	To complete this last step, we decompose $(\tilde{\bM_3} - \bM_3)$ as
	\begin{equation*}
		\tilde{\bM_3} - \bM_3  =
\left[\begin{smallmatrix} \gamma^{max} \bI_n & \bkappa \, \bX ^\intercal \\
			\bX \, \bkappa^\intercal & 2\phi \,\tilde{\bSigma} \end{smallmatrix} \right]
=
\underbrace{
\left[\begin{smallmatrix} \gamma^{max} \bI_n - \bGamma & \;\;\bzero \\
				\bzero & \;\;\bzero \end{smallmatrix}\right]
\,}_{(\mathfrak{A})} +
\underbrace{\left[ \begin{smallmatrix} \bGamma & \bkappa\, \bX ^\intercal \\
				\bX \, \bkappa^\intercal & 2\phi \,\tilde{\bSigma} \end{smallmatrix} \right]\,}_{(\mathfrak{B})},
	\end{equation*}
	where $\bGamma = \frac{1}{2\phi} \,\bkappa \,\bX^\intercal \,\tilde{\bSigma}^{-1} \,\bX \,\bkappa^\intercal$.
	Recall that $\gamma^{max}$ is the largest eigenvalue of $\bGamma$, hence $(\mathfrak{A})$ is positive semidefinite.
	It remains to prove that $(\mathfrak{B})$ is positive semidefinite as well.
	
	For any $\bw \in \mathbb{R}^{n + m}$, write $\bw = [\bw_1^\intercal, \bw_2^\intercal]^\intercal$
	where $\bw_1 \in \mathbb{R}^{n}$ and $\bw_2 \in \mathbb{R}^{m}$, then
	we have
	\begin{align*}
	& \bw^\intercal \left[ \begin{smallmatrix} \bGamma & \bkappa\, \bX ^\intercal \\ \bX \, \bkappa^\intercal & 2\phi \,\tilde{\bSigma} \end{smallmatrix} \right] \bw
	\\
	& \quad= \bw_1^\intercal \bGamma \bw_1 + 2 \bw_2^\intercal \bX \bkappa^\intercal \bw_1 + 2 \phi \bw_2^\intercal \tilde{\bSigma} \bw_2
	\\
	&\quad = \bw_1^\intercal \,\bGamma \,\bw_1 + 2 \left(\sqrt{2\phi} \,\tilde{\bsigma}\, \bw_2\right)^\intercal \left(\frac{(\tilde{\bsigma}^{-1})^\intercal}{\sqrt{2\phi}} \bX \,\bkappa^\intercal \,\bw_1\right) +  \left(\sqrt{2\phi} \,\tilde{\bsigma} \,\bw_2\right)^\intercal \left(\sqrt{2\phi} \,\tilde{\bsigma} \,\bw_2\right)
	\\
	& \quad	= \left(\frac{(\tilde{\bsigma}^{-1})^\intercal}{\sqrt{2\phi}} \,\bX\, \bkappa^\intercal \bw_1 + \sqrt{2\phi} \,\tilde{\bsigma} \,\bw_2\right)^\intercal \left(\frac{(\tilde{\bsigma}^{-1})^\intercal}{\sqrt{2\phi}} \,\bX \,\bkappa^\intercal \,\bw_1 + \sqrt{2\phi} \,\tilde{\bsigma} \,\bw_2\right)
	\\
	& \quad	 \geq  0\,.
	\end{align*}
	This implies that $(\mathfrak{B})$ is positive semidefinite and by the comparison principle of Theorem \ref{thm:comparison_riccati_equation}, the proof is complete.
	\qed
\end{proof}

%%%%%%%%%%%%%%%%%%%%%%%%%%%%%%%%%%%%%%%%%%%%%%%%%%

\subsection{Proof of Theorem \ref{thm:feynman-kac}}

To prove the result, we need to show that \eqref{eq:feynman-kac_D} is the unique solution to \eqref{eqn:PDE for D}.  To do this we introduce a sequence of approximating functions that converge to the stated solution.

Let $\Pi=\{t=t_0,t_1,\dots,t_{n_{\Pi}}=T\}$ be a partition of $[0,T]$, let $|\Pi|$ denote the cardinality of the partition $\Pi$, and let $\Delta t_k=(t_k-t_{k-1})$. Next, introduce the following piecewise (left continuous with right limits) constant approximation of $\bwC(t)\triangleq\bC^\intercal(t)\,\ba^{-1}$,
\[
\bwC^{\Pi}(t):= \sum_{k=1}^{|\Pi|} \bC^\intercal(t_{k})\, \ba^{-1} \, \mathds 1_{\{t\in(t_{k-1}, t_k]\}}\,.
\]
The time-ordered exponential of $\bwC^{\Pi}(t)$ is given by
\begin{equation}
\toe{\int_t^u \bwC^\Pi(s)\,ds} \;\; = \; e^{\bwC^{\Pi}(t_{k}) (t_{k}-t)} \left[\;\prod_{j=k+1}^{l} e^{\bwC^{\Pi}(t_{j}) \Delta t_{j}}\;\right] e^{\bwC^{\Pi}(t_{l+1}) (u-t_{l})}, \label{eqn: toe approx C}
\end{equation}
$\forall t\in[t_{k-1},t_k]$, and $u\in[t_{l},t_{l+1}]$, $l<|\Pi|$. Note that this is continuous in both $t$ and $u$ for all $t<u\in[0,T]$. We next define a sequence of functions
\begin{equation}
\bD^{\Pi}(t, \bmu) = \,  \mathbb{E}_{t,\bmu}\left[\int_t^T \toe{\int_t^u \bwC^\Pi(s)\,ds} \; \mathfrak{Z}_u \; du\right] ,
\end{equation}
where we have introduced the process $\mathfrak{Z}=(\mathfrak{Z}_t)_{t\in[0,T]}$ and
\[
\bmZ_t= \bzeta(t,\bmu_t) \quad \text{where} \quad \bzeta(t,\bmu) = 2\,\phi\;\tilde{\bSigma} \,\bmQ_t + \;\bX\;\overline{\bb}\;\bmu\,.
\]
We require the following proposition to proceed.
\begin{proposition}\textbf{PDE for approximating functions.}
The function $\bD^{\Pi}(t, \bmu)$ is the unique solution to the vector-valued PDE
\begin{equation}
\dot\bD^\Pi + \mathcal{L}^\bmu\bD^\Pi + \bwC^\Pi\,\bD^\Pi + \bzeta(t,\bmu) = \bzero^{(m)} \label{eqn:PDE for approx D}
\end{equation}
with terminal condition $\bD^\Pi(T,\bmu)=\bzero^{(m)}$.
\end{proposition}
\begin{proof}
To show this, define the stochastic process $\bmD^{\Pi}=(\bmD^{\Pi}_t)_{t\in[0,T]}$, where
\[
\bmD^{\Pi}_t = \; \toe{\int_0^t \bwC^\Pi(s)\,ds} \;\bD^{\Pi}(t, \bmu_t) + \int_0^t \toe{\int_0^u \bwC^\Pi(s)\,ds} \; \bmZ_u \; du\,.
\]
Due the Markov property of $\bmu$, we see that
\[
\bmD^{\Pi}_t =\,  \mathbb{E}\left[\left.\int_0^T \toe{\int_0^u \bwC^\Pi(s)\,ds} \; \bmZ_u \; du \; \right| \;\mathcal F^{\bmu}_t\right].
\]
By the integrability assumptions on the process $\bmu$, this is a strict martingale.
Moreover, the Markov property implies the existence of a sequence of functions $\bmd^\Pi:\mathds{R}_+\times\mathds{R}\mapsto \mathds{R} $ such that $\bmD_t^\Pi = \bmd^\Pi(t,\bmu_t)$. For any $\mathcal F^{\bmu}$-stopping time $\tau\le T$, by Dynkin's formula we have
\begin{align*}
\bzero^{(m)} =&\; \mathbb{E}[\left. \bD^{\Pi}_\tau -\bD^{\Pi}_t\;\right|\;\mathcal F^\bmu_t] \\
=&\; \mathbb{E}\left[\left. \int_t^\tau \left\{ \partial_t \bmd^\Pi(u,\bmu_u) + \mathcal L^\bmu \bmd^\Pi(u,\bmu_u) \right\}du\;\right|\;\mathcal F^\bmu_t\right]\,.
\end{align*}
Taking $\tau= (T-t) \wedge h \wedge \inf\{s \ge 0 \;:\; |\bmu_{t+s}-\bmu_t| \ge \epsilon \}$, for $h$ small, then
\begin{align}
\bzero^{(m)} =\; \mathbb{E}\left[\left. \frac{1}{h}\int_t^\tau \left\{ \partial_t \bmd^\Pi(u,\bmu_u) + \mathcal L^\bmu \bmd^\Pi(u,\bmu_u) \right\}du\;\right|\;\mathcal F^\bmu_t\right]\,.
\end{align}
As $h\downarrow0$, $\mathbb P(\tau \ne h) \downarrow 0 $, thus taking the limit as $h\downarrow0$, and using the fundamental theorem of calculus, we have
\begin{equation}
\partial_t \bmd^\Pi(t,\bmu_t) + \mathcal L^\bmu \bmd^\Pi(t,\bmu_t) = \bzero^{(m)} \,. \label{eqn: PDE bm}
\end{equation}
Furthermore, from \eqref{eqn: toe approx C},
\[
\partial_t \bmd^\Pi(t,\bmu_t) =\; \toe{\int_0^t \bwC^\Pi(s)\,ds}\;\left\{\bwC^\Pi(t)\;\bD^\Pi(t,\bmu_t) + \partial_t \bD^\Pi(t,\bmu_t) +
\bzeta(t,\bmu_t)\right\}
\]
and $\mathcal L^\bmu \bmd^\Pi(t,\bmu_t) = \; \toe{\int_0^t \bwC^\Pi(s)\,ds}\;\mathcal L^\bmu \bD^\Pi(t,\bmu_t)$, hence, as \eqref{eqn: PDE bm} holds for all paths of $\bmu$,  together with these two equalities, \eqref{eqn: PDE bm} reduces to \eqref{eqn:PDE for approx D}.
\qed
\end{proof}

Now, define the approximation error $\mathfrak{E}^{\Pi}(t,\bmu) \triangleq \bD^{\Pi}(t, \bmu)-\bD(t, \bmu)$. Taking the difference between \eqref{eqn:PDE for D} and \eqref{eqn:PDE for approx D}, we see that $\mathfrak{E}^{\Pi}$ satisfies the linear PDE
\[
\left(\partial_t + \mathcal L^\bmu\right)\mathfrak{E}^{\Pi}(t,\bmu) + \bwC^\Pi(t) \,\mathfrak{E}^{\Pi}(t,\bmu) + \left(\bwC^\Pi(t)- \bwC(t)\right)\,\bD(t,\bmu) = \bzero^{(m)}\,,
\]
with terminal condition $\mathfrak{E}^\Pi(T,\bmu)=\bzero^{(m)}$. Applying the same argument as above, $\mathfrak{E}^{\Pi}$
admits the representation
\begin{equation}
\mathfrak{E}^\Pi(t, \bmu) =\,  \mathbb{E}_{t,\bmu}\left[\int_t^T \toe{\int_t^u \bwC^\Pi(s)\,ds} \; \left(\bwC^\Pi(s)- \bwC(s)\right)\,\bD(s,\bmu_s) \; du\right]\,.
\end{equation}
It remains to show $ \mathfrak{E}^\Pi(t, \bmu) \xrightarrow{\Pi \downarrow 0}\bzero $.
By Theorem \ref{thm:bounded}, $ \bC $ is bounded and continuous on $ [0, T] $.
Therefore, by construction, $ \tilde{\bC} $, $ \tilde{\bC}^\Pi $ and $ \toe{\int_t^u \bwC^\Pi(s)\,ds} $ are all bounded and
we have $ \tilde{\bC}^\Pi \xrightarrow{\Pi \downarrow 0 } \tilde{\bC}$.
By the assumptions, there exists a constant $ C_2 > 0 $ such that
\begin{equation*}
	|D(t,\bmu)| \leq C_2 (1 + |\bmu|^2)
\end{equation*}
for all $t\in[0,T]$.
The assumptions on $ \bmu^\pm $ imply that $ \bmu $ has a finite $ \mathbb{L}^2(\Omega\times[0,T)) $-norm.
Hence $ \bD := \{\bD(t,\bmu_t)\}_{0\leq t\leq T} $ has a finite $ \mathbb{L}^1(\Omega\times[0,T)) $-norm.
The desired result follows from dominated convergence. \qed

%%%%%%%%%%%%%%%%%%%%%%%%%%%%%%%%%%%%%%%%%%%%%%%%%%

\subsection{Proof of Theorem \ref{thm:verification1}}
\label{sec:proof_of_verification_theorem}
\begin{proof}
	Under the stated assumptions, the candidate solution is indeed a classical solution of the DPE. Applying standard results (e.g., \cite{oksendal2005applied}),
	it suffices to check that (i) the SDE for $\bQ^{\bnu^*}$  has a unique solution for each given initial data;
	and (ii) $\bnu_t^*$ is indeed an admissible control.
	
	To verify (i), substituting the optimal control (\ref{eq:optimal_control})
	into the dynamics of (\ref{eq:inventory_process}),
	we have the dynamics for $\bQ_t^{\bnu^*}$
	\begin{eqnarray*}
		d\bQ_t^{\bnu^*} & = & -\tfrac{1}{2}\ba^{-1} \left( 2\,\bC(t)\,\bQ_t^{\bnu^*}+
		\left(\bE^{\intercal}(t) - \bX \right)\,\bZ_t + \bD(t,\bmu_t) \right) \,dt \,.
%		d\bZ_{t} & = &  -\bkappa \,\bZ_t \;dt + \bsigma^{\intercal}\,d\bW_{t}\,.
	\end{eqnarray*}
The above equation is an ODE with stochastic source term, and it can be explicitly integrated to find
\begin{equation}
\label{eq:optimal inventory}
\begin{split}
\bQ_t^{\bnu^*} =& \; \boldsymbol{:}\be^{-\int_0^t \ba^{-1} \bC(s)\,ds}\boldsymbol{:}\;\bQ_0 \\
& \; - \int_0^t \boldsymbol{:}\be^{-\int_u^t \ba^{-1} \bC(s)\,ds}\boldsymbol{:} \Big\{
\left(\bE^{\intercal}(u) - \bX \right)\,\bZ_u + \bD(u,\bmu_u)\Big\}\,du\,.
\end{split}
\end{equation}
Therefore, $ \bQ^{\bnu^*} $ has a unique solution for any initial data.

	To verify (ii), it suffices to show that $\bnu_t^*$ has a finite $\mathbb{L}^2(\Omega \times [0,T))$-norm.
	From (\ref{eq:optimal_control}), it suffices to show that each of $ \bZ $, $ \bQ^{\bnu^*} $ and $ \bD := \{\bD(t,\bmu_t) \}_{0\leq t\leq T} $ has a finite $\mathbb{L}^2(\Omega \times [0,T))$-norm.
		From the SDE (\ref{eq:price_dynamics}) we see that $ \bZ $ satisfies this condition.
		Moreover, from (\ref{eq:optimal inventory}) and Theorem 2 which implies that $\bC$ and $\bE$ are bounded on $[0,T]$, $\bQ $ has a finite $\mathbb{L}^2(\Omega \times [0,T))$-norm if $ \bD $ does.
	
It remains to show that $ \bD $ has a finite $\mathbb{L}^2(\Omega \times [0,T))$-norm.
From (\ref{eq:feynman-kac_D}) and the assumptions in Theorem 2, there exists a constant $ C_2 > 0 $ such that
\begin{equation*}
	|\bD(t, \bmu)| \leq C_2(1 + |\bmu|)\,,
\end{equation*}
for all $ 0 \leq t\leq T $ and $ \bmu\in \mathbb{R}^n $.
Furthermore, because the assumptions imply that $ \bmu $ has a finite $\mathbb{L}^2(\Omega \times [0,T))$-norm,
$ \bD $ also has a finite $\mathbb{L}^2(\Omega \times [0,T))$-norm, and the desired result follows. \qed
\end{proof}

\subsection{Calculating the coefficients for the limiting case}
\label{sec:detail_asy_calculation}
Substituting the power series representation (\ref{eq:asy_expansion}) into the matrix differential equations \eqref{eq:matrix_riccati_block_form},
we obtain the following equations
{\small
\begin{subequations}
\begin{align}
\bzero &= - \sum_{n=0}^\infty (n+1)\, \mathscr A_{n+1} \,\tau^n - \sum_{n=1}^\infty \left[ \bkappa\, \mathscr A_n + \mathscr A_n \,\bkappa^\intercal \right] \,\tau^n
+ \left[\tfrac{1}{4} \sum_{n=1}^\infty \mathscr E_n \, \tau^n\right]
			\ba^{-1} \left[ \sum_{n=1}^\infty \mathscr E_n^\intercal \, \tau^n \right], \label{eq:asy_equation_A}
\\[0.25em]
\bzero &= \frac{\mathscr C_{-1} }{\tau^2} - \sum_{n=0}^{\infty} (n+1)\, \mathscr C_{n+1} \,\tau^n - \phi \,\tilde{\bSigma}
+ \left[\frac{\mathscr C_{-1}^\intercal }{\tau} + \sum_{n=0}^{\infty} \,{\mathscr C_n}^\intercal\, \tau^n \right] \ba^{-1} \left[ \frac{\mathscr C_{-1} }{\tau} + \sum_{n=0}^{\infty} \,\mathscr C_n\, \tau^n \right],  \label{eq:asy_equation_C}
\\[0.25em]
\bzero &= - \sum_{n=0}^\infty \left(n+1\right)\,\mathscr E_{n+1} \,\tau^n - \sum_{n=0}^\infty \bkappa \,\mathscr E_n \,\tau^n
		    +    \left[\sum_{n=1}^\infty \mathscr E_n  \,\tau^n\right]   \ba^{-1} \left( \frac{\mathscr C_{-1} }{\tau} + \sum_{n=0}^{\infty} \left(\,\mathscr C_n \right) \, \tau^n \right).  \label{eq:asy_equation_E}%
\end{align}
\label{eq:asy_matrix_equation}
\end{subequations}}

Matching the constant terms in (\ref{eq:asy_equation_A}), we have $\mathscr A_1 = 0$.
Matching the coefficients for $\tau^{-2}$ in (\ref{eq:asy_equation_C}), we have $\mathscr C  = -\ba$.
Matching the coefficients for $\tau^{-1}$ in (\ref{eq:asy_equation_C}) yields the following equality
\begin{equation*}
	\mathscr C_{-1} \,\ba^{-1} \, \mathscr C_0 + \mathscr C_0 \, \ba^{-1} \,\mathscr C_{-1} = \bzero\,.
\end{equation*}
Therefore $\mathscr C_0 = \bzero$.

Finally, by matching the constant terms in (\ref{eq:asy_equation_E})
\[
	- \mathscr E_1 - \bkappa\, \bX^\intercal + \mathscr E_1 \,\ba^{-1}\, \mathscr C_{-1} = 0\,.
\]
This implies $\mathscr E_1 = -\frac{1}{2} \, \bkappa\, \bX^\intercal$.

It remains to show that $ \bD(t,\bmu) $ admits the asymptotic representation $ \mathcal{O} (\tau) \bmu + \mathcal{O}(\tau) $.
	From the assumptions in Theorem \ref{thm:feynman-kac}, we have $ \bE_{0,\bmu}\left[ |\bmu_t| \right] < C\,\left( 1 + |\bmu| \right) $ for all $ 0\leq t\leq T $ and some constant $ C > 0 $.
	As we assume that $ \bmu $ is Markov, we also have
	\begin{equation*}
		\bE_{t,\bmu}\left[ |\bmu_u| \right] < C\,\left( 1 + |\bmu| \right)\,,
	\end{equation*}
	for $ 0\leq t\leq u\leq T $.
	The above bound, together with \eqref{eq:feynman-kac_D}, yields
	\begin{equation*}
		|D(t,\bmu)| \leq \int_{t}^{T} C_2 + C_3\, |\bmu| \, du\,,
	\end{equation*}
	for constants $ C_2, C_3 > 0 $.
	The desired result follows.

\qed

\clearpage

\section{Parameter Estimates}

In this appendix we collect the various parameter estimates from the five Nasdaq traded stocks INTC, SMH, FARO, NTAP and ORCL.
\begin{table}[h!]
\centering
\footnotesize
\begin{tabular}{cccccc}
 & INTC & SMH & FARO & NTAP & ORCL \\ \hline \\
 $\hat \theta$ &34.233 & 51.720 & 56.338 & 43.179 & 38.885 \\ \hline \\
Co-int factor &-0.904 &  0.763 &  0.048 & -0.164 &  0.931 \\ \hline \\
%$\hat{\bb}$ & $1.09 \times 10^{-6}$ & $0.57 \times 10^{-5}$ & $2.22 \times 10^{-4}$ & $0.8 \times 10^{-5}$ & $0.32 \times 10^{-5}$ \bigstrut[t]\\
%          & ($0.54 \times 10^{-6}$) & ($0.49 \times 10^{-5}$) & ($1.17 \times 10^{-4}$) & ($0.35 \times 10^{-5}$) & ($1.14 \times 10^{-6}$) \\
    $\hat{\ba}$ & $0.44 \times 10^{-6}$ & $0.71 \times 10^{-6}$ & $0.32 \times 10^{-3}$ & $3.05 \times 10^{-6}$ & $1.35 \times 10^{-6}$ \\
    & ($2.37 \times 10^{-7}$) & ($2.58 \times 10^{-7}$) & ($1.62 \times 10^{-4}$) & ($1.27 \times 10^{-6}$) & ($0.56 \times 10^{-6}$) \\ \hline \\
%$\hat\lambda^+$ & $438.86$ & $55.89$ & $20.43$ & $249.52$ & $305.44$ \bigstrut[t]\\
%          & ($248.85$) & ($40.2$) & ($9.66$) & ($101.54$) & ($135.44$) \\
%    $\mathbb E[\eta^+]$ & $1048.99$ & $379.23$ & $100.12$ & $271.18$ & $491.11$ \\
%          & ($351.13$) & ($116.45$) & ($19.29$) & ($54.99$) & ($100.2$) \\
    $\hat\lambda^-$ & $453.91$ & $59.4$ & $21.88$ & $251.87$ & $304.13$ \\
          & ($264.63$) & ($49.46$) & ($9.25$) & ($102.72$) & ($146.83$) \\
    $\mathbb E[\eta^-]$ & $1013.83$ & $380.32$ & $98.58$ & $270.8$ & $505.59$ \\
          & ($306.58$) & ($121.39$) & ($15.78$) & ($55.2$) & ($100.29$) \bigstrut[b]\\
    \hline\end{tabular}
\caption{The first two rows (data November 3, 2014) show  mean-reverting level $\theta$ (in dollars) and weights of the co-integrating factor. The rest of table employs data for the entire year 2014. Row 3 shows the estimates of temporary price impact. We assume no cross effects so only provide the diagonal  elements of the matrix $\ba$, and also assume no permanent impact. Row 4 shows the standard deviation of the estimates in row 3.  The bottom 4 rows show the average incoming rates of MOs and their average volume: $\lambda^-$ is the average number of sell MO per hour over the year 2014, $\mathbb E[\eta^-]$ is the average volume of MOs. The standard deviation of the estimate
is shown in parentheses.}
\label{table: theta coint fact a b}
\end{table}

%\clearpage

\begin{table}[h!]
\centering
\footnotesize
\begin{tabular}{r|rrrrr}
\hline
\hline
      & INTC  & SMH   & FARO  & NTAP  & ORCL  \bigstrut\\
\hline
INTC  & 45.66 & -38.51 & -2.43 & 8.26  & -47.01 \bigstrut[t]\\
	     & (10.70) & (8.99) & (0.57) & (1.93) & (11.02) \\
SMH   & -19.83 & 16.73 & 1.06  & -3.59 & 20.42 \\
	     & (13.42) & (11.27) & (0.72) & (2.41) & (13.82) \\
FARO  & -41.34 & 34.87 & 2.20  & -7.48 & 42.57 \\
	     & (51.50) & (43.27) & (2.75) & (9.27) & (53.05) \\
NTAP  & 4.98  & -4.20 & -0.27 & 0.90  & -5.13 \\
	     & (12.17) & (10.22) & (0.65) & (2.19) & (12.54) \\
ORCL  & -6.47 & 5.45  & 0.34  & -1.17 & 6.66 \\
	     & (6.30) & (5.29) & (0.34) & (1.14) & (6.49) \bigstrut[b]\\
\hline
\hline
\end{tabular}%

\caption{Estimated mean-reverting matrix $\bkappa$ and t statistics. }
\label{table:est_kappa_5S60S}
\end{table}

\begin{table}[h!]
\centering
\footnotesize
\begin{tabular}{r|rrrrr}
\hline
\hline
      & INTC  & SMH   & FARO  & NTAP  & ORCL  \bigstrut\\
\hline
INTC  & 0.124 & 0.108 & -0.040 & 0.027 & 0.019 \bigstrut[t]\\
SMH   & 0.108 & 0.194 & 0.060 & 0.060 & 0.027 \\
FARO  & -0.040 & 0.060 & 2.855 & 0.058 & 0.001 \\
NTAP  & 0.027 & 0.060 & 0.058 & 0.159 & 0.022 \\
ORCL  & 0.020 & 0.027 & 0.001 & 0.022 & 0.043 \bigstrut[b]\\
\hline
\hline
\end{tabular}%

\caption{Estimated covariance matrix $\Sigma$.}
\label{table:est_Sigma_5S60S}
\end{table}

\begin{table}[h!]
\centering \footnotesize
\begin{tabular}{r|rr}
\hline
\hline
      & INTC  & SMH  \bigstrut\\
\hline
INTC  & 0.131 & 0.105 \bigstrut[t]\\
SMH   & 0.105 & 0.195 \bigstrut[b]\\
\hline
\hline
\end{tabular}%

\caption{Estimated covariance matrix $\bSigma^{AC}$.}
\label{table:est_Sigma_AC}
\end{table}

\begin{table}[h!]
	\centering
	\footnotesize
	\begin{tabular}{r|rrrrr}
		\hline
		\hline
		& INTC  & SMH   & FARO  & NTAP  & ORCL  \bigstrut\\
		\hline		
		INTC & 73.92 & -62.79 & -3.50 & 19.57 & -77.42 \bigstrut[t] \\
		& (10.99) & (9.34) & (0.52) & (2.93) & (11.50) \\
		SMH & 8.91 & -7.57 & -0.42 & 2.36 & -9.33 \\
		& (14.25) & (12.11) & (0.67) & (3.79) & (14.93) \\
		FARO & 48.73 & -41.39 & -2.31 & 12.90 & -51.04 \\
		& (50.80) & (43.18) & (2.40) & (13.53) & (53.20) \\
		NTAP & 11.48 & -9.75 & -0.54 & 3.04 & -12.02 \\
		& (12.25) & (10.41) & (0.58) & (3.26) & (12.83) \\
		ORCL & -15.83 & 13.45 & 0.75 & -4.19 & 16.59 \\
		& (6.39) & (5.44) & (0.30) & (1.70) & (6.70) \bigstrut[b] \\
		\hline
		\hline
	\end{tabular}%

	\caption{Estimated (with error) mean-reverting matrix $\bkappa$ and t statistics. }
	\label{table:est_kappa_5S60S_with_error}
\end{table}

\begin{table}[h!]
	\centering
	\footnotesize
	\begin{tabular}{r|rrrrr}
		\hline
		\hline
		& INTC  & SMH   & FARO  & NTAP  & ORCL  \bigstrut\\
		\hline
		INTC & 0.155 & 0.155 & -0.032 & 0.053 & 0.032 \bigstrut[t] \\
		SMH  & 0.155 & 0.260 & 0.061 & 0.084 & 0.033 \\
		FARO & -0.032 & 0.061 & 3.299 & 0.144 & -0.006 \\
		NTAP & 0.053 & 0.084 & 0.144 & 0.192 & 0.021 \\
		ORCL & 0.032 & 0.033 & -0.006 & 0.021 & 0.053 \bigstrut[b] \\
		\hline
		\hline
	\end{tabular}%

	\caption{Estimated (with error) covariance matrix $\Sigma$.}
	\label{table:est_Sigma_5S60S_with_error}
\end{table}

\begin{table}[h!]
	\centering \footnotesize
	\begin{tabular}{r|rr}
		\hline
		\hline
		& INTC  & SMH  \bigstrut\\
		\hline
		INTC  & 0.169 & 0.160 \bigstrut[t]\\
		SMH   & 0.160 & 0.261 \bigstrut[b]\\
		\hline
		\hline
	\end{tabular}%
	
	\caption{Estimated (with error) covariance matrix $\bSigma^{AC}$.}
	\label{table:est_Sigma_AC_with_error}
\end{table}

\clearpage

\section*{References}
\bibliographystyle{chicago}
\bibliography{CointAlgoTrading}

\begin{thebibliography}{}

\bibitem[\protect\citeauthoryear{Alfonsi, Fruth, and Schied}{Alfonsi
  et~al.}{2010}]{AlfonsiFruthSchiedQF2010}
Alfonsi, A., A.~Fruth, and A.~Schied (2010).
\newblock Optimal execution strategies in limit order books with general shape
  functions.
\newblock {\em Quantitative Finance\/}~{\em 10\/}(2), 143--157.

\bibitem[\protect\citeauthoryear{Almgren}{Almgren}{2003}]{Almgren03}
Almgren, R. (2003).
\newblock Optimal execution with nonlinear impact functions and
  trading-enhanced risk.
\newblock {\em Applied Mathematical Finance\/}~{\em 10\/}(1), 1--18.

\bibitem[\protect\citeauthoryear{Almgren}{Almgren}{2012}]{almgren2012optimal}
Almgren, R. (2012).
\newblock Optimal trading with stochastic liquidity and volatility.
\newblock {\em SIAM Journal on Financial Mathematics\/}~{\em 3\/}(1), 163--181.

\bibitem[\protect\citeauthoryear{Almgren}{Almgren}{2014}]{almgren2013executionFixed}
Almgren, R. (2014).
\newblock {\em High Frequency Trading; New Realities for Trades, Markets and
  Regulators}, Chapter Execution Strategies in Fixed Income Markets. Eds:
  Easley, D. and L\'opez de Prado, M. and M. O'Hara.
\newblock Risk Books.

\bibitem[\protect\citeauthoryear{Almgren and Chriss}{Almgren and
  Chriss}{2001}]{almgren2001optimal}
Almgren, R. and N.~Chriss (2001).
\newblock Optimal execution of portfolio transactions.
\newblock {\em Journal of Risk\/}~{\em 3}, 5--40.

\bibitem[\protect\citeauthoryear{Bank, Soner, and Vo{\ss}}{Bank
  et~al.}{2015}]{bank2015hedging}
Bank, P., H.~M. Soner, and M.~Vo{\ss} (2015).
\newblock Hedging with temporary price impact.
\newblock {\em Mathematics and Financial Economics\/}, 1--25.

\bibitem[\protect\citeauthoryear{Bayraktar and Ludkovski}{Bayraktar and
  Ludkovski}{2014}]{bayraktar2012liquidation}
Bayraktar, E. and M.~Ludkovski (2014, October).
\newblock Liquidation in limit order books with controlled intensity.
\newblock {\em Mathematical Finance\/}~{\em 24\/}(4), 627--650.

\bibitem[\protect\citeauthoryear{Cartea, Donnelly, and Jaimungal}{Cartea
  et~al.}{2013}]{CDJ}
Cartea, {\'A}., R.~Donnelly, and S.~Jaimungal (2013).
\newblock Algorithmic trading with model uncertainty.
\newblock {\em SSRN: http://ssrn.com/abstract=2310645\/}.

\bibitem[\protect\citeauthoryear{Cartea and Jaimungal}{Cartea and
  Jaimungal}{2015}]{cartea2014optimal}
Cartea, {\'A}. and S.~Jaimungal (2015).
\newblock Optimal execution with limit and market orders.
\newblock {\em Quantitative Finance\/}~{\em 15\/}(8), 1279--1291.

\bibitem[\protect\citeauthoryear{Cartea and Jaimungal}{Cartea and
  Jaimungal}{2016a}]{CarJaiCointegration}
Cartea, {\'A}. and S.~Jaimungal (2016a).
\newblock Algorithmic trading of co-integrated assets.
\newblock {\em International Journal of Theoretical and Applied Finance\/}~{\em
  19\/}(06), 1650038.

\bibitem[\protect\citeauthoryear{Cartea and Jaimungal}{Cartea and
  Jaimungal}{2016b}]{cartea2014closed}
Cartea, {\'A}. and S.~Jaimungal (2016b).
\newblock A closed-form execution strategy to target volume weighted average
  price.
\newblock {\em SIAM Journal on Financial Mathematics\/}~{\em 7\/}(1), 760--785.

\bibitem[\protect\citeauthoryear{Cartea and Jaimungal}{Cartea and
  Jaimungal}{2016c}]{cartea2015incorporating}
Cartea, {\'A}. and S.~Jaimungal (2016c).
\newblock Incorporating order-flow into optimal execution.
\newblock {\em Mathematics and Financial Economics\/}~{\em 10\/}(3), 339--364.

\bibitem[\protect\citeauthoryear{Cartea, Jaimungal, and Kinzebulatov}{Cartea
  et~al.}{2016}]{AlSebDam13}
Cartea, {\'A}., S.~Jaimungal, and D.~Kinzebulatov (2016).
\newblock Algorithmic trading with learning.
\newblock {\em International Journal of Theoretical and Applied Finance\/}~{\em
  19\/}(04), 1650028.

\bibitem[\protect\citeauthoryear{Cartea, Jaimungal, and Penalva}{Cartea
  et~al.}{2015}]{TheBook}
Cartea, {\'A}., S.~Jaimungal, and J.~Penalva (2015).
\newblock {\em Algorithmic and High-Frequency Trading\/} (1st ed.).
\newblock Cambridge: Cambridge University Press.

\bibitem[\protect\citeauthoryear{Forsyth, Kennedy, Tse, and Windcliff}{Forsyth
  et~al.}{2012}]{forsyth2012optimal}
Forsyth, P.~A., J.~S. Kennedy, S.~Tse, and H.~Windcliff (2012).
\newblock Optimal trade execution: a mean quadratic variation approach.
\newblock {\em Journal of Economic Dynamics and Control\/}~{\em 36\/}(12),
  1971--1991.

\bibitem[\protect\citeauthoryear{G\^arleanu and Pedersen}{G\^arleanu and
  Pedersen}{2013}]{JOFI:JOFI12080}
G\^arleanu, N. and L.~Pedersen (2013).
\newblock Dynamic trading with predictable returns and transaction costs.
\newblock {\em The Journal of Finance\/}~{\em 68\/}(6), 2309--2340.

\bibitem[\protect\citeauthoryear{Gatheral, Schied, and Slynko}{Gatheral
  et~al.}{2012}]{SchiedMAFI2012}
Gatheral, J., A.~Schied, and A.~Slynko (2012).
\newblock Transient linear price impact and {F}redholm integral equations.
\newblock {\em Mathematical Finance\/}~{\em 22\/}(3), 445--474.

\bibitem[\protect\citeauthoryear{Gu{\'e}ant and Lehalle}{Gu{\'e}ant and
  Lehalle}{2015}]{GueantLeHalleMF13}
Gu{\'e}ant, O. and C.-A. Lehalle (2015).
\newblock General intensity shapes in optimal liquidation.
\newblock {\em Mathematical Finance\/}~{\em 25\/}(3), 457--495.

\bibitem[\protect\citeauthoryear{Gu\'eant, Lehalle, and
  Fernandez~Tapia}{Gu\'eant et~al.}{2012}]{Gueant2012}
Gu\'eant, O., C.-A. Lehalle, and J.~Fernandez~Tapia (2012).
\newblock Optimal portfolio liquidation with limit orders.
\newblock {\em SIAM Journal on Financial Mathematics\/}~{\em 3\/}(1), 740--764.

\bibitem[\protect\citeauthoryear{Guilbaud and Pham}{Guilbaud and
  Pham}{2013}]{guilbaud2013optimal}
Guilbaud, F. and H.~Pham (2013).
\newblock Optimal high-frequency trading with limit and market orders.
\newblock {\em Quantitative Finance\/}~{\em 13\/}(1), 79--94.

\bibitem[\protect\citeauthoryear{Jaimungal and Kinzebulatov}{Jaimungal and
  Kinzebulatov}{2013}]{jaimungal2013optimal}
Jaimungal, S. and D.~Kinzebulatov (2013).
\newblock Optimal execution with a price limiter.
\newblock {\em Available at SSRN 2199889\/}.

\bibitem[\protect\citeauthoryear{Kharroubi and Pham}{Kharroubi and
  Pham}{2010a}]{kharroubi2010optimal}
Kharroubi, I. and H.~Pham (2010a).
\newblock Optimal portfolio liquidation with execution cost and risk.
\newblock {\em SIAM Journal on Financial Mathematics\/}~{\em 1\/}(1), 897--931.

\bibitem[\protect\citeauthoryear{Kharroubi and Pham}{Kharroubi and
  Pham}{2010b}]{KharroubiPhamSIAM2010}
Kharroubi, I. and H.~Pham (2010b).
\newblock Optimal portfolio liquidation with execution cost and risk.
\newblock {\em SIAM Journal on Financial Mathematics\/}~{\em 1\/}(1), 897--931.

\bibitem[\protect\citeauthoryear{Kratz}{Kratz}{2011}]{kratz2011optimal}
Kratz, D.-M.~P. (2011).
\newblock {\em Optimal liquidation in dark pools in discrete and continuous
  time}.
\newblock Ph.\ D. thesis, Humboldt-Universit{\"a}t zu Berlin.

\bibitem[\protect\citeauthoryear{Lei and Xu}{Lei and Xu}{2015}]{lei2015costly}
Lei, Y. and J.~Xu (2015).
\newblock Costly arbitrage through pairs trading.
\newblock {\em Journal of Economic Dynamics and Control\/}~{\em 56}, 1--19.

\bibitem[\protect\citeauthoryear{Leung and Li}{Leung and
  Li}{2015}]{leung2015optimal}
Leung, T. and X.~Li (2015).
\newblock Optimal mean reversion trading with transaction costs and stop-loss
  exit.
\newblock {\em International Journal of Theoretical and Applied Finance\/},
  1550020.

\bibitem[\protect\citeauthoryear{Lintilhac and Tourin}{Lintilhac and
  Tourin}{2016}]{lintilhac2016model}
Lintilhac, P. and A.~Tourin (2016).
\newblock Model-based pairs trading in the bitcoin markets.
\newblock {\em Quantitative Finance\/}, 1--14.

\bibitem[\protect\citeauthoryear{Mudchanatongsuk, Primbs, and
  Wong}{Mudchanatongsuk et~al.}{2008}]{mudchanatongsuk2008optimal}
Mudchanatongsuk, S., J.~A. Primbs, and W.~Wong (2008).
\newblock Optimal pairs trading: A stochastic control approach.
\newblock In {\em American Control Conference, 2008}, pp.\  1035--1039. IEEE.

\bibitem[\protect\citeauthoryear{Ngo and Pham}{Ngo and
  Pham}{2016}]{ngo2014optimal}
Ngo, M.-M. and H.~Pham (2016).
\newblock Optimal switching for the pairs trading rule: A viscosity solutions
  approach.
\newblock {\em Journal of Mathematical Analysis and Applications\/}~{\em
  441\/}(1), 403 -- 425.

\bibitem[\protect\citeauthoryear{{\O}ksendal and Sulem}{{\O}ksendal and
  Sulem}{2005}]{oksendal2005applied}
{\O}ksendal, B.~K. and A.~Sulem (2005).
\newblock {\em Applied stochastic control of jump diffusions}, Volume 498.
\newblock Springer.

\bibitem[\protect\citeauthoryear{Passerini and Vazquez}{Passerini and
  Vazquez}{2016}]{AlphaPredictors}
Passerini, F. and S.~Vazquez (2016).
\newblock Optimal trading with alpha predictors.
\newblock {\em Journal of Investment Strategies\/}~{\em 5(3)}, 2047--1238.

\bibitem[\protect\citeauthoryear{Schied}{Schied}{2013}]{SchiedAMF13}
Schied, A. (2013).
\newblock Robust strategies for optimal order execution in the
  {Almgren---Chriss} framework.
\newblock {\em Applied Mathematical Finance\/}~{\em 20\/}(3), 264--286.

\bibitem[\protect\citeauthoryear{Tourin and Yan}{Tourin and
  Yan}{2013}]{yan2012dynamic}
Tourin, A. and R.~Yan (2013).
\newblock Dynamic pairs trading using the stochastic control approach.
\newblock {\em Journal of Economic Dynamics and Control\/}~{\em 37\/}(10), 1972
  -- 1981.

\end{thebibliography}

\end{document}